\documentclass[journal,twoside]{IEEEtran}
\usepackage{graphicx}
\usepackage{epstopdf}
\usepackage{placeins}
\usepackage{array}
\DeclareGraphicsExtensions{.eps,.pdf} 
\graphicspath{ {figures/} }
\usepackage{cite}
\usepackage{balance}
\usepackage{amssymb,amsmath}
\usepackage{subfigure}
\usepackage[ruled,vlined]{algorithm2e}
\usepackage{amssymb,amsmath,mathtools}
\usepackage{amssymb,amsmath}
\usepackage{algorithmic}
\usepackage{color}
\usepackage{multirow}
\usepackage{stackrel,amssymb,amsmath}
\usepackage{empheq}
\usepackage{cases}
\makeatletter
\newcommand\restartchapters{\par
  \setcounter{chapter}{0}%
  \setcounter{section}{0}%
  \gdef\@chapapp{\chaptername}%
  \gdef\thechapter{\@arabic\c@chapter}}
\makeatother

\def\bq{{\bf q}}

\def\bw{{\bf w}}

\newtheorem{remark}{\it \underline{Remark}}
\newtheorem{definition}{\bf Definition}
\newtheorem{lemma}{\it \underline{Lemma}}

\newtheorem{proposition}{\it \underline{Proposition}}

\renewcommand{\algorithmicrequire}{\textbf{Input:}}

\makeatletter
\g@addto@macro\normalsize{%
 \setlength\abovedisplayskip{4pt}
 \setlength\belowdisplayskip{4pt}
 \setlength\abovedisplayshortskip{4pt}
 \setlength\belowdisplayshortskip{4pt}
}
\makeatother

\makeatletter
\def\endthebibliography{%
	\def\@noitemerr{\@latex@warning{Empty `thebibliography' environment}}%
	\endlist
}
\makeatother

\begin{document}
\bstctlcite{IEEEexample:BSTcontrol}
\title{UAV Relay-Assisted Emergency Communications in  IoT Networks: Resource Allocation  and Trajectory Optimization}
\author{\IEEEauthorblockN{ $\text{Dinh-Hieu Tran}, \textit{Graduate Student Member, IEEE}$, $\text{Van-Dinh Nguyen}, \textit{Member, IEEE}$, $\text{Symeon Chatzinotas}, \textit{Senior Member, IEEE},$ \\  $\text{Thang X. Vu}, \textit{Member, IEEE},$  $\text{and Bj{\"o}rn Ottersten}, \textit{Fellow, IEEE}$  }
\thanks{This work is supported in part by the Luxembourg National Research Fund under project FNR CORE ProCAST, grant C17/IS/11691338 and FNR 5G-Sky, grant C19/IS/13713801.}
\thanks{Dinh-Hieu Tran, V.-D.~Nguyen, Symeon Chatzinotas, Thang X. Vu, and Bj{\"o}rn Ottersten are with the Interdisciplinary Centre for Security, Reliability and Trust (SnT), the University of Luxembourg, Luxembourg. (e-mail: \{hieu.tran-dinh, dinh.nguyen,  symeon.chatzinotas, thang.vu, bjorn.ottersten \} @uni.lu).}
\thanks{Part of this work was presented in IEEE GLOBECOM 2020 - Workshop on Future of Wireless Access for Industrial IoT (FutureIIoT), Taipei, Taiwan \cite{Hieuglobecom2020}.}
\thanks{}
\thanks{}
 }
\maketitle

\begin{abstract} 
Unmanned aerial vehicle (UAV) communication has emerged as a prominent technology for emergency communications (e.g., natural disaster) in the Internet of Things (IoT) networks to enhance the ability of disaster prediction, damage assessment, and rescue operations promptly. A UAV can be deployed as a flying base station (BS) to collect data from time-constrained IoT devices and then transfer it to a ground gateway (GW). In general, the latency constraint at IoT devices and UAV's limited storage capacity highly hinder practical applications of UAV-assisted IoT networks. In this paper, {full-duplex (FD) radio} is adopted at the UAV to overcome these challenges. In addition, half-duplex (HD) scheme for UAV-based relaying is also considered to provide a comparative study between two modes (viz., FD and HD). {Herein, a device is considered to be successfully served iff its data is collected by the UAV and conveyed to GW timely during flight time}. In this context, we aim to maximize the number of served IoT devices by jointly optimizing bandwidth, power allocation, and the UAV trajectory while satisfying each device's requirement and the UAV's limited storage capacity. The formulated optimization problem is troublesome to solve due to its non-convexity and combinatorial nature. {Towards appealing applications, we first relax binary variables into continuous ones and transform the original problem into a more computationally tractable form.} By leveraging inner approximation framework, we derive newly approximated functions for non-convex parts and then develop a simple yet efficient iterative algorithm for its solutions. Next, we attempt to maximize the total throughput subject to the number of served IoT devices. Finally, numerical results show that the proposed algorithms significantly outperform benchmark approaches in terms of the number of served IoT devices and system throughput.
\end{abstract}

\begin{IEEEkeywords}	
Emergency communications, full-duplex, information freshness, Internet-of-Things, timely data collection, unmanned aerial vehicle (UAV).
\end{IEEEkeywords}

\section{Introduction} \label{Introduction}
{In 1999, British technology pioneer Kevin Ashton introduced the concept of the Internet-of-Things (IoT) to describe a system in which all devices equipped with sensors can connect to each other \cite{ashton2009}.} IoT has the potential to significantly enhance the quality of human life such as smart home, health care, wearable devices, agriculture, smart city, autonomous vehicles, and smart grid \cite{QuocViet2019,Phu2021MEC}. The number of IoT connections of all types is estimated to reach close to 25 billion by 2025 \cite{ericsson2019}. However, the growing demand for communications is becoming a major challenge for IoT networks due to limited spectral resources at terrestrial base stations (BSs). {Besides, BSs are deployed at fixed locations} and antenna height to serve a fixed geographical area, and resources cannot be rapidly shifted elsewhere. Especially in emergency communications, whereas BSs are potentially isolated or damaged after a natural disaster, or when BSs are unable to serve all users as they are overloaded during peak hours. This raises a question of how to support the communication needs of a massive number of IoT devices with restricted resources without compromising the network performance \cite{yan2018}. Fortunately, due to the high maneuverability and flexible deployment, unmanned aerial vehicle (UAV) communications could become a promising technology to overcome {the above mentioned shortcomings \cite{Yzeng2019accessing}. Due to energy constraints, IoT devices are commonly unable to propagate their signals long distances. On the other hand,} the UAV can fly closer to the devices, harvest the IoT data, and then transmit it to the BS/control center, which is out of the transmission range of these devices.

Extensive studies have been carried out to investigate UAV-assisted IoT communication networks \cite{mozaffari2016,feng2018,motlagh2019,wang2019,yuan2018,Liu2019UAVIoT}. The work in \cite{mozaffari2016} studied the joint optimal 3D deployment of UAVs, uplink (UL) power control, and device association in an IoT network. Specifically, the authors proposed a new framework for efficiently distributing UAVs to collect information in the UL from IoT users. {In \cite{feng2018}, the authors optimized the data gathering efficiency of a UAV-assisted IoT network, subject to the power budget, energy capacity, and total transmission time of IoT} devices. Herein, a multi-antenna UAV was operated, which followed a circular trajectory and served IoT devices to create a virtual multi-input multi-output (MIMO) channel. Reference \cite{motlagh2019} presented a robust central system orchestrator (SO) that was designed to provide value-added IoT services (VAIoTS). Whereas SO keeps the entire details about UAVs including their current locations, flight missions, total energy budget, and their onboard IoT devices. To obtain an efficient UAV selection mechanism corresponding to each task requirement, the authors proposed three solutions, namely, energy-aware UAV, fair trade-off UAV, and delay-aware UAV selection. A novel UAV-aided IoT communication network to provide energy-efficient data gathering and accurate 3D device positioning of IoT devices was proposed in \cite{wang2019}, whereas a UAV was deployed as an aerial anchor node and a flying data collector. {Particularly, UAVs could serve not only as aerial BSs but also as powerful IoT components that are capable of performing communications, sensing, and data analysis while hovering in the air \cite{yuan2018}. To extend the coverage for IoT-based emergency communications, Liu et al. \cite{Liu2019UAVIoT} integrated multi-hop device-to-device (D2D) and UAV communication during natural disasters so that helps out-of-range IoT users can be effectively connected to UAVs. Note that none of the above-mentioned works in \cite{mozaffari2016,feng2018,motlagh2019,wang2019,yuan2018,Liu2019UAVIoT} take crucial latency constraint into consideration.}

Recently, the delay-sensitive data collection has attracted much attention from researchers \cite{liu2018,abd2018,li2019,Samir}. For example, in the emergency case or during the natural disaster, the out-of-date gathering data may result in unreliable controllable decisions, which may ultimately be disastrous \cite{Samir}. On the other hand, IoT devices often have limited storage capacity, and thus their generated data need to be collected timely before it becomes worthless due to obsolete transmissions or being overwritten by incoming data \cite{schulz2017}. Therefore, the UAV must reach the right place at the right time. In \cite{liu2018}, the authors proposed two UAV trajectories, termed Max-AoI-optimal and Ave-AoI-optimal, to efficiently collect data from ground sensor nodes under the impact of age of information (AoI) metric. Specifically, the Max-AoI-optimal and Ave-AoI-optimal trajectory planning minimize the age of the oldest information and the average AoI of all sensor nodes, respectively. The work in \cite{abd2018} studied the role of a UAV acting as a relay to minimize the average Peak AoI for a transmitter-receiver link, which was accomplished via a joint optimization of the UAV trajectory, energy spending, and service time allocations for packet transmissions. In \cite{li2019}, the authors designed the UAV trajectory to minimize expired data packets in UAV-enabled wireless sensor networks (WSNs) and then applied the reinforcement learning (RL) method for the solution, which enhances the time-effectiveness and path design performance. The authors in \cite{Samir} optimized the UAV trajectory as well as service bandwidth allocation to maximize the total number of served ground IoT users, in which UAV needs to collect data from users within their latency constraint. Different from  \cite{liu2018,abd2018,li2019,Samir}, which only studied the aspect of data collection on the UL channel, the works in \cite{tran2019} and \cite{Hieu} further considered the latency constraint on the DL channel.

{Despite noticeable achievements for data collection in UAV-assisted IoT networks \cite{mozaffari2016,feng2018,motlagh2019,wang2019,yuan2018,liu2018,abd2018,li2019,Samir}, aforementioned works have not exploited benefits of FD radios.} To efficiently exploit the radio spectrum, FD transmission was adopted in UAV communications \cite{song2018,ZhuHan2018,duo2020,ye2020}. By applying a circular trajectory and decode-and-forward (DF) relaying strategy, the work in \cite{song2018} maximized instantaneous data rate by a joint design of beam-forming and power allocation, under individual and sum-power constraint for the source and relay users. In \cite{ZhuHan2018}, the authors investigated the spectrum sharing planning problem for FD UAV relaying systems with underlaid device-to-device (D2D) communications, which aims to maximize the sum throughput. {The work in \cite{duo2020} maximized the energy efficiency (EE) by jointly optimizing UAV trajectory, as the transmit and jamming powers of a source and a UAV, respectively.} Besides, a new system model for UAV-enabled FD wireless-powered IoT networks was proposed in \cite{ye2020}, in which three optimization problems, namely, total-time minimization, sum-throughput maximization, and total energy minimization problem, were investigated.

{Unlike previous studies such as \cite{liu2018,abd2018,li2019, Samir,tran2019, Hieu} that only investigate timely data exchange on the UL or DL channel utilizing HD mode, this work proposes a novel system model in UAV relay-assisted IoT networks that further explores the impact of requested timeout (RT) constraints for both UL and DL transmissions.} To the best of our knowledge, this is the first work to jointly optimizes total bandwidth, transmission power, trajectory design, storage capacity, and latency constraint in UAV relay-assisted IoT networks. To this end, we formulate two optimization problems and develop efficient iterative algorithms to obtain a sub-optimal solution. In summary, our contributions are as follows:
\begin{itemize}
	\item We propose a novel UAV relay-assisted IoT model that takes into account the latency requirement for UL and DL channels to improve the freshness of information. Therein, UAV-enabled FD relaying is exploited as an effective mean to enhance network performance, i.e., increasing the number of served IoT devices, throughput, and reducing latency.  For instance, the reduced latency and high throughput owing to FD operation can take the virtual/augmented reality (VR/AR) experiences or emergency communications to the next level. Besides, it also helps to overcome UAV's limited storage capacity. Moreover, UAV-enabled HD relaying is also investigated to fully capitalize on UAV benefits for time-sensitive data collection in IoT networks.
	
	\item We formulate a generalized optimization problem to maximize the total number of served IoT devices under the UAV's maximum speed constraint, total traveling time constant, maximum transmit power of devices/UAV, limited cache size of UAV, and latency constraints for both UL and DL. The formulation belongs to the difficult class of mixed-integer non-convex optimization problem, which is generally NP-hard. We first relax binary variables into continuous ones and penalize the objective by introducing a penalty function. We then develop an iterative computational procedure for its solutions, which guarantees convergence to at least a local optimal. The key idea behind our approach is to derive newly approximated functions for non-convex parts by employing the inner approximation (IA) framework \cite{marks}.
	
	\item Inspired by the practical requirement in human safety measurements, the more data we have collected, the better our predictions are. This motivates us to investigate the optimization problem in order to maximize the total collected throughput subject to a given number of served IoT users.

	\item The proposed schemes' effectiveness is revealed via numerical results, which show significant improvements in both number of served IoT devices and the total amount of collected throughput compared with the benchmarks. More specifically, the Benchmark FD and Benchmark HD schemes are respectively designed similar to the proposed FD-based and HD-based methods but with fixed resource allocation or fixed trajectory.
	
	\item {Compared to our conference \cite{Hieuglobecom2020}, we have made the following major revisions.} Firstly, the work in \cite{Hieuglobecom2020}  only considers the throughput maximization problem with an assumption of perfect CSI from IoT devices to UAV. Moreover, the details of mathematical analysis are not provided in \cite{Hieuglobecom2020}. In this manuscript, we have updated the channel model considering the approximated rate functions for both uplink and downlink, as given in Lemma 1. Besides, we have provided the IA framework in Section III and detailed the proof of Proposition 1 in Appendix D. We have added an efficient method to generate an initial feasible point to start the IA-based algorithm in Section III-B. Lastly, we have reproduced all simulation results in Section V due to the change of channel model. In addition, we have also added Fig. 3 to illustrate the UAV's trajectories.

\end{itemize}

\begin{figure}[t]
	\centering
	\includegraphics[width=9cm,height=8cm]{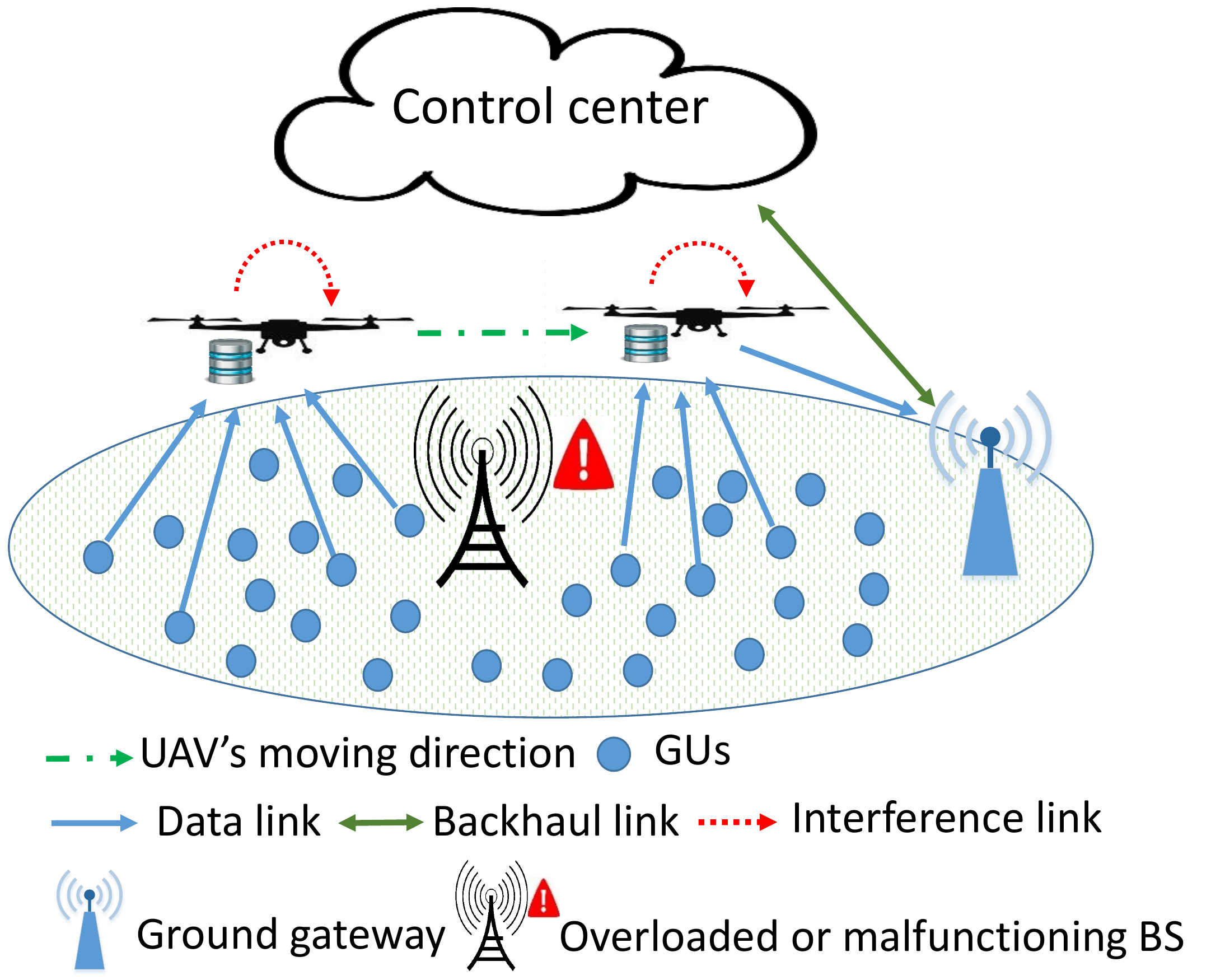}
	\caption {System model: the UAV is deployed as a flying BS to collect the data from IoT devices and then transmit to GW. }
	\label{fig:1}   
\end{figure}
\begin{figure}[t]
	\centering
	\includegraphics[width=9cm,height=6cm]{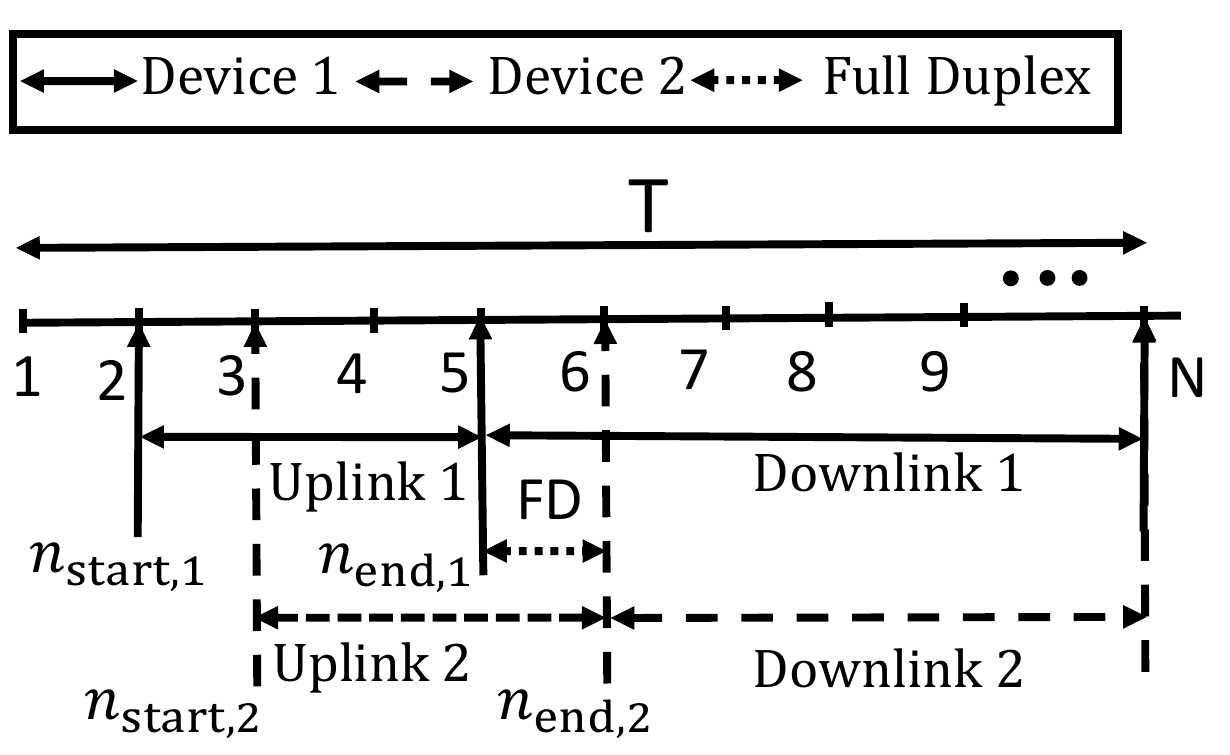}
	\caption{ Illustration of the data transmission process of 2 IoT devices with $N$ time intervals. The first IoT device with initial data transmission time at $n_{{\rm start},1}=2$, timeout at $n_{{\rm end},1}=5$. The second IoT device with initial data transmission time at $n_{{\rm start},2}=3$, timeout at $n_{{\rm end},2}=6$. The UAV operates in the FD mode from time slots 5 to 6 since two devices utilize the same sub-carrier. }
	\label{fig:2}   
\end{figure}

The rest of the paper is organized as follows. The system model and problem formulation are given in Section~\ref{System Model}. The proposed iterative algorithm for FD is presented in Section~\ref{Sec:3}. Section~\ref{sec:HD} devotes for the HD scheme. Numerical results are illustrated in Section~\ref{Sec:Num}, and Section~\ref{Sec:Conclusion} concludes the paper.

\emph{Notation}: Scalars and vectors are denoted by lower-case letters and boldface lower-case letters, respectively. For a set $\mathcal{{K}}$, $|\mathcal{{K}}|$ denotes its cardinality. For a vector $v$, $\left\| v \right\|_1$ and $\left\| v \right\|$  denote its $\ell_1$ and Euclidean ($\ell_2$) norm, respectively. $\mathbb{R}$ represents for the real matrix. $\mathbb{R}^+$ denotes the non-negative real numbers, i.e., $\mathbb{R}^+=\{x \in \mathbb{R}|x \ge 0\}$. $x \sim {\cal{CN}}(0,\sigma^2)$ represents circularly symmetric complex Gaussian random variable with zero mean and variance $\sigma^2$. Finally, $\triangledown f$ is the first derivative of a function $f$. $\mathbb{E} [x]$ denotes the expected value of $x$.


\section{System Model and  Problem Formulation} \label{System Model}
We consider a UAV-aided cooperative wireless IoT network, where a UAV is deployed to assist the existing terrestrial communication infrastructure in the case of adverse conditions or natural calamities, as shown in Fig. \ref{fig:1}. In emergency communications, the ground base station (GBS) is either partially or completely damaged after a natural disaster or in the case that the GBS is overloaded during the peak hours due to its incapability of handling all the devices at the same time (e.g., a sporting event) \cite{Zeng}. The latter case has been recognized as one of the key scenarios that need to be effectively solved by fifth-generation (5G) wireless communication \cite{Osseiran,Zeng}. Concretely, a UAV helps to relay data from a set of $K$ IoT devices (or GUs), denoted by $\mathcal{K}\triangleq \{1,\ldots,K\}$, to a GW. Each IoT device is equipped with a single antenna and works in HD mode. Due to the SWAP (size, weight, and power) limitations, the UAV, acting as an on-demand relay, is equipped with one FD antenna, which can simultaneously be used for data transmission and reception. Specifically, the UAV can operate in FD or HD mode depending on the system designer. It hovers over the considered area to effectively gather data from IoT devices and then transmit it to the GW using UL and DL communications, respectively. {Due to limited energy budget, we restrain the total serving time of UAV as $T$ \cite{Sun2019SolarUAV}}. We assume that each device is active at different time instances $t$, where $0 \le t \le T$. The location of device $k$ is denoted as ${\bw_k} \in {\mathbb{R}^{2 \times 1}}, k \in \mathcal{K}$. We assume that the locations of IoT devices together with their data sizes, the initial data transmission time (i.e., $n_{{\rm start},k}$ with $k\in {\cal K}$), and latency requirement (i.e., $n_{{\rm end},k}$ with $k\in {\cal K}$) are known to the UAV through the control center.\footnote{The control center can take care of the corresponding computations and inform the UAV through dedicated signaling, without affecting the performance of the considered framework.} Denote $n_{{\rm start},k}$ and $n_{{\rm end},k}$ by the initial data transmission time and timeout constraint of the device $k$, respectively, for $k \in \cal{K}$. It is assumed that the UAV collects data from device $k$ within $n_{{\rm end},k}$ units of time. For simplicity, we assume that the UAV flies at a constant altitude of $H$ (m), e.g., imposed by the regulatory authority for safety considerations. The location of UAV projected on the ground at time t is denoted as $\bq(t) \in {\mathbb{R}^{2 \times 1}},$ with $0 \le t \le T$ \cite{Y_Zeng_1}.

{\subsection{UAV-to-Ground and Ground-to-UAV Channel Model}}
\label{Device_UAV_model}

For ease of exposition, the time horizon $T$ is discretized into $N$ equally spaced time intervals, i.e., $T=N \delta _t$ with $\delta _t$ being the primary slot length. {Moreover, let ${\cal N}=\{1,\dots,N\}$ denote a set of all time slots.} Note that the UAV location can be assumed to be approximately unchanged during each time slot compared to the distance from the UAV to IoT devices since $\delta _t$ is chosen sufficiently small \cite{Cheng2019UAVCache}. Then, the UAV trajectory $\bq(t)$ during time horizon $T$ can be represented as $\left( {\bq[n ]} \right)_{n  = 1}^N, $ where $\bq[n ]$ denotes the UAV's horizontal location at $n$-th time interval. Let $V_{\rm max}$ denote the maximum velocity of the UAV, then the UAV's speed constraint can be presented as
\begin{align}
\label{eq:1}
\left\| {\bq[n ] - \bq[n  - 1]} \right\|\le \delta_d = V_{\rm max}{\delta _t},n  = 2,...,N.
\end{align}

For notation convenience, let us denote the $k$-th IoT device and UAV by $k$ and $\rm U$, respectively. Henceforth, $1k$ and $2k$ represent for the UL (i.e., $k \to \rm U$) and DL (i.e., ${\rm U} \to \rm GW$), respectively. Then, the time-dependence distance from $k \to \rm U$ or ${\rm U} \to \rm GW$ (i.e., $1k$ or $2k$), is given by
\begin{align}
\label{eq:2}
{d_{ik}}[n] = \sqrt {{H^2} + {\left\| {\bq[n] - \bw} \right\|}^2} , i \in \{1,2\}, \; \forall n, \; k,
\end{align}
where $\bw \in \{\bw_k,\bw_0\}$, with $\bw_0$ denoting the location of GW.

In realistic scenarios, the devices are located in different environments, e.g., rural, urban, suburban, etc. Thus, a generalized channel model consisting of both line-of-sight (LOS) and non-line-of-sight (NLOS) channel elements is considered. In this work, we consider a practical channel model that takes into account both large-scale and small-scale fading channels \cite{Yaxiong}. Specifically, the channel coefficient at the $n$-th time slot, denoted by $h_{ik}[n]$, can be written as \cite{gong,Samir}
\begin{align}
\label{eq:6}
h_{ik}[n] = \sqrt {\omega_{ik}[n]} {\tilde h}_{ik}[n],
\end{align}
where $\omega_{ik}[n]$ represents for the large-scale fading effects and ${\tilde h}_{ik}[n]$ accounts for Rician small-scale fading coefficient. Specifically, $\omega_{ik}[n]$ can be modeled as
\begin{align}
\label{eq:2_11}
\omega_{ik}[n] = \omega_0 d_{ik}^{-\alpha}[n],
\end{align}
where $\omega_0$ is the average channel power gain at the reference distance $d=$ 1 m, and $\alpha \ge 2$ is the path loss exponent for the Rician fading channel \cite{Samir}. The small scale fading ${\tilde h}_{ik}[n]$ with an expected value $\mathbb{E} \left[ | {\tilde h}_{ik}[n]  |^2 \right]=1$, is given by
\begin{align}
\label{eq:2_12}
{\tilde h}_{ik}[n] = \sqrt{\frac{G}{1+G}} \bar{h}_{ik}[n] + \sqrt{\frac{1}{1+G}} \hat{h}_{ik}[n],
\end{align}
where $G$ is the Rician factor; $\bar{h}_{ik}[n]$ and $\hat{h}_{ik}[n] \sim \mathcal{CN}(0,1)$ denote the deterministic LoS and the NLoS component (Rayleigh fading) during time slot $n$, respectively.

Due to the UL and DL channels' coexistence using the same frequency at $n$-th time slot, the self-interference (SI) {may occur at the UAV. Without loss of generality, once the UAV finishes data collection from device $k$, then the transmission from UAV to GW can be conducted.}\footnote{In this work, we adopt a (decode-and-forward) DF relaying technique \cite{hieu2018performance}; thus, the UAV needs to complete receiving all the data from device $k$ before relaying to GW to guarantee the data encoding properly. Moreover, a sufficiently large time period is assumed to carry out the data transfer as well as the decoding process at the UAV.} 

Let us denote by $x_{1k}[n]$ and $x_{2k}[n]$ the data symbols with unit power (i.e., $\mathbb{E} \left[|x_{1k}[n]|^2\right]=1$ and $\mathbb{E} \left[|x_{2k}[n]|^2\right]=1$) from $k \to \rm U$ and $\rm U \to \rm GW$ at time slot $n$, respectively. As a result, the received signals of device $k$ at the UAV and GW are respectively given by
\begin{align}
	\label{eq:8}
y_{1k} [n]& = \sqrt{p_{1k}[n]} h_{1k}[n]x_{1k}[n] \notag\\&+  \sqrt{\rho^{\rm RSI}}g_{\rm U}[n] \sum\limits_{k^\ast \in {\cal K} \setminus k} \sqrt{p_{2k^{\ast}}[n]} x_{2k^{\ast}}[n] + n_0, \\
\label{eq:9}
y_{2k}[n] &=  \sqrt{p_{2k}[n]}h_{2k}[n]x_{2k}[n] + n_0,
\end{align}
where RSI represents for residual self-interference term, {$\sqrt{\rho^{\rm RSI}}g_{\rm U}[n] \sum\limits_{k^\ast \in {\cal K} \setminus k} \sqrt{p_{2k^{\ast}}[n]} x_{2k^{\ast}}[n]$ is the RSI power after all interference cancellations \cite{Dinh,Sabharwal,Dinh_1,Nguyen}, $\rho^{\rm RSI}\in [0,1)$ is the degree of RSI,} $n_0 \sim {\cal{CN}}(0,\sigma^2)$ denotes the additive white Gaussian noise (AWGN); $p_{1k}[n]$ and $p_{2k}[n]$ are the transmit power of the device $k$ and UAV on the UL and DL to transmit the device $k$'s data at time slot $n$, respectively; $g_{\rm U}[n]$ denotes the fading loop channel at the UAV, which interferes UL reception due to concurrent downlink transmission \cite{Duarte,Dan}.

To deal with the issues involved in limited resources and the UAV's self-interference, we consider the resources allocation (i.e., bandwidth and transmit power) for bold the UL and DL. Thus, the achievable rate (bits/s) of links from $k \to {\rm U}$ or ${\rm U} \to {\rm GW}$ to transmit the data of device $k$ at time slot $n$ are respectively given as
\begin{align}
\label{eq:10}
r_{ik}[n] &= a_{ik}[n] B \log_2 \left(1+\Gamma_{ik}\right), i \in \{1,2\},
\end{align}
where $\Gamma_{1k} \triangleq \frac{ p_{1k}[n] | {\tilde h}_{1k}[n] |^2 \omega_0 } { {\left({H^2} + {{\left\| {\bq[n] - \bw_k} \right\|}^2}\right)^{\alpha/2}} \big(\phi^{\rm RSI} \sum\limits_{k^\ast \in {\cal K} \setminus k} p_{2k^\ast}[n] +  \sigma^2\big)}$, $\Gamma_{2k} \triangleq \frac{ p_{2k}[n] {| {\tilde h}_{2k}[n] |^2 \omega_0 }}{\left({H^2} + {{\left\| {\bq[n] - \bw_0} \right\|}^2}\right)^{\alpha/2} \sigma^2}$, $\phi^{\rm RSI} \triangleq \rho^{\rm RSI}|g_{\rm U}[n]|^2$; $B$ denotes the total bandwidth in hertz (Hz) of the system; $a_{1k}[n] B$ and $a_{2k}[n] B$ are the bandwidth allocated for the UL and DL to transmit data of $k$-th device during time slot $n$, respectively. Herein, $a_{1k}[n]$ and $a_{2k}[n]$ represent for the spectrum allocation for devices and the UAV, respectively. Note that instantaneous CSI elements (i.e., $| {\tilde h}_{1k}[n] |^2$ and $| {\tilde h}_{2k}[n] |^2$) are difficult to obtain in advance. Moreover, $| {\tilde h}_{1k}[n] |^2$ and $| {\tilde h}_{2k}[n] |^2$ are random variables, thus instantaneous rates (i.e., $r_{1k}[n]$ and $r_{2k}[n]$) are also random variables. {Therefore, the expected values of received rates at the UAV/GW are expressed as \cite{Hua_UAV_Back2020,HieuBack2020}
\begin{align}
	\label{eq:rate1}
	\mathbb{E}\big[r_{ik}[n]\big] &= a_{ik}[n]B \mathbb{E}[ \log_2 \big(1+\Gamma_{ik}\big)], i \in \{1,2\}, k \in {\cal K}.
\end{align}}

Due to the troublesome of deriving the probability density function, it raises a difficulty in obtaining the closed-form expression of $\mathbb{E}\big[r_{ik}[n]\big]$. {Thus, we provide lower-bound functions of $\mathbb{E}\big[r_{ik}[n]\big]$ as follows:
\begin{lemma}\label{lemma:1}
	The lower bounds of $\mathbb{E}\big[r_{1k}[n]\big]$ and $\mathbb{E}\big[r_{2k}[n]\big]$ are respectively given as 
	\begin{align}
		\label{eq:Lemma1_1}
		\bar{r}_{1k}[n] &= a_{1k}[n] B \log_2 \Bigg(1+\frac{e^{-E} p_{1k}[n]  \omega_0 } { {({H^2} + {{\left\| {\bq[n] - \bw_k} \right\|}^2})^{\alpha/2}} \nu_{1k}[n] }\Bigg), \\
		\label{eq:Lemma1_2}
		\bar{r}_{2k}[n] &= a_{2k}[n] B \log_2 \left(1+\frac{e^{-E} p_{2k}[n] { \omega_0 }}{({H^2} + {{\left\| {\bq[n] - \bw_0} \right\|}^2})^{\alpha/2} \sigma^2}\right),
	\end{align}
where $\nu_{1k}[n] \triangleq \big(\phi^{\rm RSI} \sum\limits_{k^\ast \in {\cal K} \setminus k} p_{2k^\ast}[n] +  \sigma^2\big)$.
\begin{proof}
	See Appendix~A.
\end{proof}
\end{lemma}}

Further, $a_{1k}[n]$ and $a_{2k}[n]$ are approximately continuous between 0 and 1 \cite{Samir}. Thus, the bandwidth allocation should satisfy: 
\begin{align}
\label{eq:101}
 &\sum\limits_{k \in {\cal K}} a_{ik}[n] \le 1, \forall n,  \; i \in \{1,2\},\\
 \label{eq:102}
&  0 \le a_{ik}[n] \le 1, \forall k,n.
\end{align}

Based on \eqref{eq:Lemma1_1} and \eqref{eq:Lemma1_2}, the throughput (in bits) received on the UL or DL to transmit device $k$'s data during time slot $n$, can be written as
\begin{align}
\label{eq:15}
C_{ik} [n] =
\delta_t R_{ik}[n], \; {\rm where} \; i \in \{1,2\},
\end{align}
where
\begin{align}
\label{eq:17}
R_{ik}[n] &=  \left\{  \begin{array}{l} 
\bar{r}_{ik}[n],\hfill \text{if} \; n \in {\cal T}_{ik}, \;\\
0,\hfill \text{otherwise},
\end{array} \right.
\end{align}
with ${\cal T}_{1k}\triangleq\{n_{{\rm start},k},\ldots,n_{{\rm end},k}\}$, ${\cal T}_{2k} \triangleq \{n_{{\rm end},k}+1,\ldots, N\}$; \eqref{eq:17} means that the UAV only can collect the data from device $k$ (or transmit data to the GW) during time period ${\cal T}_{1k}$ (or ${\cal T}_{2k}$); Otherwise, the data transmission rate is treated as zero. Specifically, the UAV only transmits device $k$'s data to GW iff it finishes the data collection process for that device. Moreover, the total throughput over $N$ time slots received on the UL and DL are denoted as $C_{1k}=\sum\limits_{n \in {\cal T}_{1k} } \delta_t R_{1k} [n]$ and $C_{2k}=\sum\limits_{n \in {\cal T}_{2k}} \delta_t R_{2k}[n]$, respectively.

To assist in the mathematical problem formulation, we introduce a new binary variable $\lambda_k$ as
\begin{align}
	\label{eq:24}
	{\lambda_k} = \left\{ \begin{array}{l}
		1, \hfill {\rm{Device}} \; k \;{\rm{ is \; successfully\; served \; by \; the \; UAV}}, \\
		0, \hfill {\rm{otherwise.}}
	\end{array} \right.
\end{align}
\begin{definition}
	The value of $\lambda_k$ should be equal to one iff the device $k$'s data is collected by the UAV while additionally guaranteeing its successful reception at the GW.
\end{definition} 

Let $S_k$ denote the data size (in bits) needed to transmit from device $k$ to GW. Then, we have the RT constraint for transmitting the device $k$'s data on the UL and DL are expressed as, respectively
\begin{align}
\label{eq:19} 
&\lambda_k \frac{ S_k}{R_{1k}} \le (n_{{\rm end},k}-n_{{\rm start},k}+1)\delta_t, \forall k,\\
\label{eq:20} 
&\lambda_k \frac{ S_k}{R_{2k}} \le (N-n_{{\rm end},k})\delta_t, \forall k,
\end{align}
where $R_{ik}=\sum\limits_{n \in {\cal T}_{ik} } R_{ik}[n]$, with $i \in \{1,2\}$; \eqref{eq:19} means that device $k$ must transmit information to the UAV before timeout constraint, i.e., $n_{{\rm end},k}$; \eqref{eq:20} implies that the data transmission process to transmit the device $k$'s data from $\rm U \to \rm GW$ is performed during the serving time of the UAV.\footnote{We consider the system model in which the UAV does not transmit the data during taking off and landing \cite{Challita}. Thus, the data transmission process only occurs when the UAV is flying in the sky.}

\subsection{Caching Model}
\label{Caching_model}

The UAV has a cache with a storage capacity of $C$. Due to the limited cache size of the UAV, it can utilize FD mode to release the storage and improve the network throughput. Considering a storage limitation, the total number of files cached at the UAV should not exceed its available storage capacity, i.e.,
\begin{align}
\label{eq:21}
\sum\limits_{k \in {\cal K}} \Bigg(\sum\limits_{l=1}^n C_{1k}[l] -  \sum\limits_{l=1}^{n-1} C_{2k}[l] \Bigg)  \le C,
\end{align}
where $\sum\limits_{l=1}^n C_{1k}[l] \triangleq \lambda_k S_k - \sum\limits_{l=n+1}^N  C_{1k}[l]$.

Note that, in order to spend a part of storage capacity for future use, i.e., a free cache size to store new data streams, the amount of data stored at the UAV is calculated as the size of files collected from all devices till $n$-th time slot minus files transmitted to GW till $(n-1)$-th time slot as in \eqref{eq:21}.

\subsection{Problem Formulation}

In this section, we aim to maximize the total number of served IoT devices by jointly optimizing the UAV trajectory ${\bq}[n]$, the allocation of resources (i.e., bandwidth and transmit power assigned for UL and DL), and taking into account the storage limitation, the locations, initial transmission time, and the timeout constraint of all IoT devices. 

Let us define $\bq \triangleq \{\bq[n], \forall n\}$, ${\bf a}\triangleq \{a_{1k}[n], a_{2k}[n], k \in {\cal K}, n \in {\cal N} \},$ ${\bf p} \triangleq \{p_{1k}[n], p_{2k}[n], k \in {\cal K},  n \in {\cal N} \},$ ${ \boldsymbol \lambda} \triangleq \{{\lambda_k}, k \in {\cal K}\}$. Based on the above discussions, the problem for maximizing number of served IoT devices can be mathematically formulated as follows:
\begin{IEEEeqnarray}{rCl}\label{eq:P1}
{\cal P}^{\rm FD}:\ &&\max_{\bq, {\bf a}, {\bf p}, {\boldsymbol \lambda} }~~ \left\|  {\boldsymbol \lambda}  \right\|_1  \IEEEyessubnumber \label{eq:P1:a}\\
\mathtt{s.t.}~~
\vspace{-0.01cm}
&&{\lambda_k \in \{0,1\}, \forall k,} \IEEEyessubnumber\label{eq:P1:b}\\
\vspace{-0.01cm}
&&\delta_t \; {{\rm min} (R_{1k},R_{2k})  \ge \lambda_k S_k, \forall k,} \IEEEyessubnumber\label{eq:P1:c}
\\
\vspace{-0.01cm}
&&\lambda_k \frac{ S_k}{R_{1k}} \le (n_{{\rm end},k}-n_{{\rm start},k}+1)\delta_t, \forall k, \IEEEyessubnumber\label{eq:P1:e}\\
\vspace{-0.01cm}
&&\lambda_k \frac{ S_k}{R_{2k}} \le (N-n_{{\rm end},k})\delta_t, \forall k, \IEEEyessubnumber\label{eq:P1:f}\\
\vspace{-0.02cm}
&&\sum\limits_{k \in {\cal K}} \left(\lambda_k S_k -  \sum\limits_{l=n+1}^N \delta_t R_{1k}[l] -   \sum\limits_{l=1}^{n-1} \delta_t R_{2k}[l] \right)  \le C, \notag\\ \hfill \forall n, \IEEEyessubnumber\label{eq:P1:g} \\
\vspace{-0.01cm}
&& \sum\limits_{k \in {\cal K}} a_{ik}[n] \le 1, \forall n, \IEEEyessubnumber\label{eq:P1:h} \\
\vspace{0.01cm}
&&  0 \le a_{ik}[n] \le 1, \forall k,n, \IEEEyessubnumber\label{eq:P1:i}\\
\vspace{-0.01cm} &&\bq[1]= \bq_{\rm I}, \bq[N]=\bq_{\rm F}, \IEEEyessubnumber\label{eq:P1:j} \\ \vspace{-0.01cm}
&&\left\| {\bq[n]-\bq[n-1]} \right\| \le \delta_d, n=2,\dots,N,  \IEEEyessubnumber\label{eq:P1:k} \\ \vspace{-0.01cm}
&&0 \le   p_{1k}[n] \le P_k^{\rm max}[n], \forall k,n,  \IEEEyessubnumber\label{eq:P1:l} \\ 
&&0 \le \sum\limits_{k \in {\cal K}}  p_{2k}[n] \le P_{\rm U}^{\rm max}[n], \forall n, \IEEEyessubnumber\label{eq:P1:m} 
\end{IEEEeqnarray} 
where constraint \eqref{eq:P1:c} means that each IoT device needs to upload an amount of data $S_k$. In constraint \eqref{eq:P1:j}, $\bq_{\rm I}$ and $\bq_{\rm F} \in \mathbb{R}^{2 \times 1}$ denote the beginning and ending locations of UAV projected onto horizontal plane, respectively; \eqref{eq:P1:k} signifies the maximum speed constraint of the UAV; constraints \eqref{eq:P1:l} and \eqref{eq:P1:m} imply maximum transmit power constraints.

The problem ${\cal P}^{\rm FD}$ is a mixed integer non-linear program (MINLP), which is generally NP-hard. Moreover, the binary constraint \eqref{eq:P1:b} and non-convex constraints \eqref{eq:P1:c} to \eqref{eq:P1:g} cause intractability. Therefore, it is cumbersome to find an efficient solution of ${\cal P}^{\rm FD}$ directly. However, a suitable solution (local or global optimal) may be obtained by employing adequate relaxations to ${\cal P}^{\rm FD}$. In this regard, we provide a transformation mechanism for ${\cal P}^{\rm FD}$, followed by its corresponding solution in the succeeding section.

\section{Proposed Iterative Algorithm for solving ${\cal P}^{\rm FD}$}
\label{Sec:3}
This section provides an iterative algorithm based on the IA method to solve the design problem. The principle of IA framework \cite{marks,beck2010} can be detailed as follows. Let us consider the following non-convex problem
\begin{IEEEeqnarray}{rCl}
	\label{IP}
 &&\min_{x \in \mathbb{R}^n}~~ f(x) \IEEEyessubnumber\label{IP:a}\\
	\mathtt{s.t.}~~
	&& h_m(x) \le 0, m = \{1,\dots,M\},  \IEEEyessubnumber\label{IP:b}
\end{IEEEeqnarray}
where $f(x)$ and $h_m(x)$ are non-convex and continuous differentiable functions over $\mathbb{R}^n$. The key idea of IA is to replace  non-convex functions by its approximated convex ones. Let us denote $\bar{f}^{(j)}(x)$ and $\bar{h}_m^{(j)}(x), \forall m$ are new convex functions,  satisfying the properties listed in \cite{beck2010}, i.e., 
\begin{align}
	f(x) &\le \bar{f}^{(j)}(x) \; \text{and} \; h_m(x) \le \bar{h}_m^{(j)}(x), \\
	f(x^{(j)}) &= \bar{f}^{(j)}(x^{(j)}) \; \text{and} \; h_m(x^{(j)}) = \bar{h}_m^{(j)}(x^{(j)}), \\
	\frac{\partial f(x)}{\partial x}\bigg|_{x=x^{(j)}} &= \frac{\partial \bar{f}^{(j)}(x)}{\partial x}\bigg|_{x=x^{(j)}} \; \notag\\ &\text{and} \;	\frac{\partial h_m(x)}{\partial x}\bigg|_{x=x^{(j)}} = \frac{\partial \bar{h}_m^{(j)}(x)}{\partial x}\bigg|_{x=x^{(j)}},
\end{align}
where $x^{(j)} \in {\mathcal{F}} \triangleq \{ x| {\rm s.t.} \; \eqref{IP:b} \}$ is a feasible point of problem \eqref{IP} at iteration $j$. In some special cases, the approximated functions $(\bar{f}^{(j)}(x),\bar{h}_m^{(j)}(x))$ can be easily obtained by adopting first-order Taylor approximation of $(f(x),h_m(x))$ at feasible point $x^{(j)}$. Consequently, we solve the approximate convex program at iteration $j$ of an iterative algorithm, which is given by
\begin{IEEEeqnarray}{rCl}
	\label{IP:convex}
	&&\min_{x \in \mathbb{R}^n}~~ \bar{f}^{(j)}(x) \IEEEyessubnumber\label{IP:convex:a}\\
	\mathtt{s.t.}~~
	&& \bar{h}^{(j)}_m(x) \le 0, m = \{1,\dots,M\}. \IEEEyessubnumber\label{IP:convex:b}
\end{IEEEeqnarray}
A general iterative algorithm to solve \eqref{IP:convex} is presented as follows: $i$) Generate the initial feasible point $x^{(0)}\in {\mathcal{F}}$; $ii$) At  iteration $j$, the optimal solution $x^\star$ is obtained by solving \eqref{IP:convex}; $iii$) Update $x^{(j+1)} \triangleq x^\star$ and $j=j+1$; $iv$) Repeats steps $(ii)-(iii)$ until convergence. The detailed proof of convergence can be found in  \cite{marks,beck2010}.
\subsection{Tractable Formulation for \eqref{eq:P1}}
In this section, we aim to make problem $({\cal P}^{\rm FD})$ more tractable by relaxing the binary variables of \eqref{eq:P1:b} into continuous values, i.e.,  $0 \le \lambda_k \le 1$. {To obtain near-exact binary solutions at optimum, we introduce the penalty function to penalize uncertainties of the binary nature. It is straightforward to see that ${\lambda _k} \in \left\{ {0,1} \right\} \Leftrightarrow 		\big(	0 \le {\lambda _k} \le 1 \; \& \;
{\lambda _k} - ( {{\lambda _k}} )^2 \le 0 \big).$ We see that the convex function $\mathbb{P}(\boldsymbol{\lambda}) \triangleq \sum\limits_{k \in {\cal K}} \lambda_k (\lambda_k -1)$ with $0 \le \lambda_k \le 1, \forall k$ is always non-positive and can be used to measure the degree of satisfaction of \eqref{eq:P1:b}. Similar to \cite{CheJoint2014,HDTPMU}, instead of handling the non-convex constraint ${\lambda _k} - ( {{\lambda _k}} )^2 \le 0$, we maximize the penalty function $\mathbb{P}(\boldsymbol{\lambda})$ to achieve its satisfaction by incorporating it in the objective function (see, e.g., \cite[Chapter 16]{bonnans2006numerical}).} Hence, the parameterized relaxed problem with penalty parameter $\mu \in \mathbb{R}^+$ is expressed as
\begin{IEEEeqnarray}{rCl}
	\vspace{-0.01cm}
	\label{eq:P11}
	{{\cal P}^{\rm FD}_{\rm relaxed}}: &&\max_{\bq, {\bf a}, {\bf p}, {\boldsymbol \lambda} }~~ \left\|  {\boldsymbol \lambda} \right\|_1 + \mu \mathbb{P}(\boldsymbol{\lambda}) \IEEEyessubnumber\label{eq:P11:a}\\
	\mathtt{s.t.}~~ \vspace{-0.01cm}
	&&0 \le \lambda_k \le 1, 
	\forall k,  \IEEEyessubnumber\label{eq:P11:b}\\ \vspace{-0.01cm}
	&&  {\eqref{eq:P1:c}-\eqref{eq:P1:m}}.
	\IEEEyessubnumber\label{eq:P11:c}
\end{IEEEeqnarray}

{\begin{remark}
		Note that in the parameterized relaxed problem ${\cal P}^{\rm FD}_{\rm relaxed}$ \eqref{eq:P11},	the binary variables in the original problem \eqref{eq:P1} are relaxed to continuous ones between 0 and 1. Therefore, if $\lambda_k, \forall k$ are all binary at optimal, then the relaxation is tight and the obtained solution is also a feasible solution of problem \eqref{eq:P1}. Theoretically, ${\mathbb{P}}(\boldsymbol{\lambda})$ should be zero at convergence to guarantee the same objective value with \eqref{eq:P1} under the sufficiently large value of $\mu$. Nevertheless, there exists a numerical tolerance in computation and it can be accepted if $\mathbb{P}(\boldsymbol{\lambda})  < \epsilon$, where $\epsilon$ is a very small chosen value corresponding to a large value of $\mu$ \cite{TungEnergy2018,CheJoint2014,Tungspectral2018}.
	\end{remark}
}

However, a direct application of IA method to solve ${\cal P}_{\rm relaxed}^{\rm FD}$ is inapplicable due to  non-concavity of the objective function and non-convexity of constraints in \eqref{eq:P1:c}-\eqref{eq:P1:g} as well as strong coupling among optimization variables. In what follows, we transform \eqref{eq:P11} into an equivalent non-convex problem where the IA method can be applied. In this context, we introduce slack variables $z_{1k}[n]$, $z_{2k}[n]$, and $t_{1k}[n]$ such that $\bigl({H^2} + {\left\| {\bq[n] - \bw_k} \right\|}^2\bigr) \le (z_{1k}[n])^{2/\alpha}$,  $\bigl({H^2} + {\left\| {\bq[n] - \bw_0} \right\|}^2\bigr) \le (z_{2k}[n])^{2/\alpha}$, and $\phi^{\rm RSI} \sum\limits_{k^\ast \in {\cal K} \setminus k} p_{2k^\ast}[n] +  \sigma^2  \le  t_{1k}[n]$, respectively, where $\alpha \ge 2$ for Rician fading channel \cite{Samir,rappaport,abhayawardhana}, by which \eqref{eq:Lemma1_1} and \eqref{eq:Lemma1_2} can be rewritten as
{\begin{align}
\label{eq:30} 
&\bar{r}_{1k}[n] \ge  r_{1k}^{\rm lb}[n] \triangleq a_{1k}[n] B \log_2 \Bigl(1+\frac{e^{-E} p_{1k}[n]  \omega_0 } {z_{1k}[n] t_{1k}[n] }\Bigr),\\
\vspace{-0.01cm}
\label{eq:31}
&\bar{r}_{2k}[n] \ge r_{2k}^{\rm lb}[n] \triangleq a_{2k}[n] B
\log_2 \Bigl(1+\frac{e^{-E} p_{2k}[n]  \omega_0 } {z_{2k}[n]\sigma^2}\Bigr).
\end{align}}

By substituting \eqref{eq:30} and \eqref{eq:31} into \eqref{eq:15} and \eqref{eq:17}, we respectively obtain $C^{\rm lb}_{ik}[n]$ and  $R^{\rm lb}_{ik}[n]$, with $i \in \{1,2\}$. Moreover, we have $R^{\rm lb}_{ik}=\sum\limits_{n \in {\cal T}_{ik} } R^{\rm lb}_{ik}[n]$ and $C_{ik}^{\rm lb}=\sum\limits_{n \in {\cal T}_{ik} } \delta_t R^{\rm lb}_{ik}[n]$. Let us denote ${\bf z}=\{z_{1k}[n], z_{2k}[n], n \in {\cal N}, k \in {\cal K} \},$ $ {\bf t} =\{t_{1k} [n],  k \in {\cal K}, n \in {\cal N} \}$. Then, the problem ${\cal P}_{\rm relaxed}^{\rm FD}$ can be reformulated as
\begin{IEEEeqnarray}{rCl}
\label{eq:P12}
{\cal P}_{\rm \rm relaxed-1}^{\rm FD}:\ &&\max_{\bq, {\bf a},  {\bf p}, {\boldsymbol \lambda}, {\bf z}, {\bf t} }~~ \left\|  {\boldsymbol \lambda} \right\|_1 + \mu \mathbb{P}(\boldsymbol{\lambda}) \IEEEyessubnumber\label{eq:P12:a}\\
\mathtt{s.t.}~~ \vspace{-0.01cm}
&&  \eqref{eq:P11:b}, \eqref{eq:P1:h}- \eqref{eq:P1:m},  
\IEEEyessubnumber\label{eq:P12:b}\\ \vspace{-0.01cm}
&& {H^2} + {{\left\| {\bq[n] - \bw_k} \right\|}^2} \le \left(z_{1k}[n]\right)^{2/\alpha},  \forall k, n, \notag\\ \vspace{-0.01cm} && {H^2} + {{\left\| {\bq[n] - \bw_0} \right\|}^2} \le \left(z_{2k}[n]\right)^{2/\alpha}, \forall n,
 \IEEEyessubnumber\label{eq:P12:c}\\ \vspace{-0.01cm}
&& \phi^{\rm RSI} \sum\limits_{k^\ast \in {\cal K} \setminus k} p_{2k^\ast}[n] +  \sigma^2 \le  t_{1k}[n],  \forall k, n, \IEEEyessubnumber\label{eq:P12:d}\\
&&\lambda_k \frac{ S_k}{R_{1k}^{\rm lb}} \le (n_{{\rm end},k}-n_{{\rm start},k}+1)\delta_t, \forall k, \IEEEyessubnumber\label{eq:P12:e}\\ \vspace{-0.01cm}
&&\lambda_k \frac{ S_k}{R_{2k}^{\rm lb}} \le (N-n_{{\rm end},k})\delta_t, \forall k, \IEEEyessubnumber\label{eq:P12:f}\\ \vspace{-0.01cm}
&&\delta_t \;{{\rm min} (R_{1k}^{\rm lb},R_{2k}^{\rm lb})  \ge \lambda_k S_k, \forall k,} \IEEEyessubnumber\label{eq:P12:g}\\ \vspace{-0.01cm}
&& \sum\limits_{k \in {\cal K}} \delta_t R_{2k}^{\rm lb} \ge \sum\limits_{k \in {\cal K}} {\lambda_k S_k}, \forall k \in {\cal{K}}, \IEEEyessubnumber\label{eq:P12:h} \\ 
&&\sum\limits_{k \in {\cal K}} \Big(\lambda_k S_k - \sum\limits_{l=n+1}^N  \delta_t R_{1k}[l] -  \sum\limits_{l=1}^{n-1} \delta_t R_{2k}[l] \Big)  \notag\\  && \le C,  \forall k, n. \IEEEyessubnumber\label{eq:P12:i}
\end{IEEEeqnarray} 

It is noteworthy that ${\cal P}_{\rm \rm relaxed-1}^{\rm FD} $ is a much simpler form in comparison to ${\cal P}^{\rm FD} $, but the possibility of a direct solution still seems unviable. This is due to the fact that joint computation of the optimization parameters (related to \eqref{eq:P12:e}-\eqref{eq:P12:i}) leads to non-convexity of the problem. However, it is still possible to solve the problem in an iterative manner. In the following, we discuss the above-mentioned approach in details.

\subsection{Proposed IA-based Algorithm}
\vspace{-0.01cm}
\label{sec:3b}
{We are now in position to convexify \eqref{eq:P12} by applying the IA method \cite{marks} under which the non-convex parts are completely exposed.}

\underline{\textit{Approximation of the objective function:}} The objective \eqref{eq:P12:a} is a convex function in $\boldsymbol{\lambda}$, which is useful to apply the IA method. In particular, the convex function $\mathbb{P}(\boldsymbol{\lambda})$ is iteratively replaced by the  linear function  $\hat{\mathbb{P}}^{(j)}(\boldsymbol{\lambda})$:
\vspace{-0.01cm}
{\begin{align}
		\label{eq:46}
		\widehat{\mathbb{P}}^{(j)}(\boldsymbol{\lambda}) &\triangleq  \mathbb{P}(\boldsymbol{\lambda}^{(j)}) + \triangledown \mathbb{P}(\boldsymbol{\lambda}^{(j)}) \big( \boldsymbol{\lambda} - \boldsymbol{\lambda}^{(j)} \big)  \notag\\
		&=\sum\limits_{k \in {\cal K}} \Big( \lambda_k (2\lambda_k^{(j)}-1) - (\lambda_k^{(j)})^2 \Big),
	\end{align}
	where $\mathbb{P}(\boldsymbol{\lambda}^{(j)})= \widehat{\mathbb{P}}^{(j)}(\boldsymbol{\lambda}^{(j)})$.
}{As a result, the objective function in problem ${\cal P}_{\rm \rm relaxed-1}^{\rm FD}$ can be replaced by $\left\|  {\boldsymbol \lambda} \right\|_1 + \mu \widehat{\mathbb{P}}^{(j)}(\boldsymbol{\lambda})$.} 

\vspace{0.2cm}

\underline{\textit{Approximation of $r_{1k}^{\rm lb}[n]$ and $r_{2k}^{\rm lb}[n]$:}} {Before proceeding further, 
we can express $r_{ik}^{\rm lb}[n]$, $i \in \{1,2\}$ as}
\vspace{-0.01cm}
{\begin{align}
    \label{eq:27_1}
	r_{ik}^{\rm lb}[n]= a_{ik}[n]  \Phi_{ik}[n],
\end{align}}
where
\vspace{-0.01cm}
{\begin{align}
\label{eq:29}
\Phi_{1k}[n] &\triangleq  B  \log_2 \left(1+\frac{e^{-E} p_{1k}[n] \omega_0 } {z_{1k}[n] t_{1k}[n]}\right), \\
\vspace{-0.01cm} \label{eq:30_1}
\Phi_{2k}[n] &\triangleq  B 
\log_2 \left(1+\frac{e^{-E} p_{2k}[n]  \omega_0}{z_{2k}[n]\sigma^2}\right). 
\end{align}}

\vspace{-0.01cm}
To approximate  \eqref{eq:29} and
\eqref{eq:30_1}, we first introduce the following lemmas:
\begin{lemma}\label{lemma:11}
	Consider a concave function $h(x, y) \triangleq \sqrt{xy}, \; x > 0, \; y > 0$. Its convex upper bound at given points $x^{(j)}$ and $y^{(j)}$ can be given by \cite[Appendix B]{Dinh}, \cite{beck2010}:
	\begin{align}
	\label{eq:Lemma11}
	h(x, y) \le \frac{\sqrt{x^{(j)}}}{2\sqrt{y^{(j)}}}y + \frac{\sqrt{y^{(j)}}}{2\sqrt{x^{(j)}}}x.
	\end{align}
\end{lemma}
\begin{lemma}\label{lemma:2}
Consider a function $h_1(x, y, z) \triangleq \ln \left(1 + \frac{x}{yz} \right)$ and $h_2(x, z) \triangleq  \ln \left(1 + \frac{x}{z} \right), x>0,\; y>0,\; z>0$. The concave lower bound of $h_1(x, y, z)$ and $h_2(x, z)$ at given point $x^{(j)}$, $y^{(j)}$, and $z^{(j)}$ are expressed as 
\vspace{-0.01cm}
\begin{align}
\label{eq:Lemma2_5}
h_1(x, y, z) &\ge \ln \left(1+\frac{x^{(j)}}{y^{(j)} z^{(j)}}\right) - \frac{x^{(j)}}{y^{(j)} z^{(j)}}  \notag\\ & + 2 \frac{\sqrt{x^{(j)}}\sqrt{x}} {y^{(j)} z^{(j)}} - \frac{x^{(j)}\left(x + \frac{y^{(j)}}{2 z^{(j)}} z^2+\frac{z^{(j)}}{2 y^{(j)}} y^2 \right)}{y^{(j)} z^{(j)}\left(x^{(j)} + y^{(j)} z^{(j)} \right) },\\
\label{eq:Lemma2_6}
 h_2(x, z) &\ge \ln \left(1+\frac{x^{(j)}}{ z^{(j)}}\right) - \frac{x^{(j)}}{ z^{(j)}}  + 2 \frac{\sqrt{x^{(j)}}\sqrt{x}} { z^{(j)}} \notag\\ &- \frac{x^{(j)}\left(x + z\right)}{ z^{(j)}\left(x^{(j)}  + z^{(j)} \right) }.
\end{align}
\begin{proof}
	See Appendix~B.
\end{proof}
\end{lemma}

{Based on Lemmas \ref{lemma:11} and \ref{lemma:2}, $\Phi_{1k}[n]$ and $\Phi_{1k}[n]$ are lower bounded by
\vspace{-0.01cm}
\begin{align}
\label{eq:34}
\Phi_{1k}[n] &\ge \bar{\Phi}_{1k}[n] \triangleq B  \big( \Xi_1 + \Xi_2-  \Xi_3\big), \\
\vspace{-0.01cm}
\label{eq:35}
\Phi_{2k}[n] &\ge \bar{\Phi}_{2k}[n] \triangleq B \big( \Xi_4 + \Xi_5-  \Xi_6 \big),
\end{align} 	
where $\Xi_1, \Xi_2, \Xi_3, \Xi_4, \Xi_5$, and $\Xi_6$ are defined in Appendix~C.} By introducing slack variable $\Phi_{ik}^{\rm lb}[n]$, $i \in \{1,2\}$, {with}
\vspace{-0.01cm}
\begin{align}
\label{eq:362}
\bar{\Phi}_{ik}[n]  & \ge \Phi_{ik}^{\rm lb}[n],
\end{align} 
\vspace{-0.01cm}{we rewrite $r_{ik}^{\rm lb}[n]$ as}
\vspace{-0.01cm}
\begin{align} \label{eq:38}
r_{ik}^{\rm lb}[n]  \ge \bar{r}_{ik}^{\rm lb}[n] \triangleq a_{ik}[n]  \Phi_{ik}^{\rm lb}[n].
\end{align} 

\vspace{-0.01cm}
{To tackle non-convex function $a_{ik}[n]\Phi_{ik}^{\rm lb}[n]$ we replace $a_{ik}[n]\Phi_{ik}^{\rm lb}[n]$ by equivalent Difference of Convex (DC) function $0.25\big[(a_{ik}[n]+\Phi_{ik}^{\rm lb}[n])^2-(a_{ik}[n]-\Phi_{ik}^{\rm lb}[n])^2\big]$.} Then, we apply the first-order Taylor approximation to approximate the convex function $(a_{ik}[n]+\Phi_{ik}[n])^2$ at the $(j+1)$-th iteration:
\begin{IEEEeqnarray}{rCl} 
	\label{eq:36}
	a_{ik}[n]\Phi_{ik}^{\rm lb}[n] &\ge& \frac{\big(a_{ik}^{(j)}[n]+\Phi_{ik}^{{\rm lb},(j)}[n]\big)^2}{4} + \frac{\big(a_{ik}^{(j)}[n]+\Phi_{ik}^{{\rm lb},(j)}[n]\big)}{2} \notag\\ &\times& \Big(a_{ik}[n] -a_{ik}^{(j)}[n]   + \Phi_{ik}^{\rm lb}[n]-\Phi_{ik}^{{\rm lb},(j)}[n]\Big) \notag \\ &-& \frac{\big(a_{ik}[n]-\Phi_{ik}^{\rm lb}[n]\big)^2}{4} \triangleq \tilde{r}_{ik}^{\rm lb}[n].
\end{IEEEeqnarray} 

\vspace{-0.01cm}
To convexify \eqref{eq:P12:e}-\eqref{eq:P12:i}, we introduce the slack variables $\widehat{r}_{ik}^{\rm lb}[n]$, with $i \in \{1,2\}$, to equivalently express \eqref{eq:36} as
\begin{IEEEeqnarray}{rCl} 
	\label{eq:44}
	&& \tilde{r}_{ik}^{\rm lb}[n]  \ge \widehat{r}_{ik}^{\rm lb}[n], i \in \{1,2\}.
\end{IEEEeqnarray}
As a result, substituting $\widehat{r}_{ik}^{\rm lb}[n]$ into \eqref{eq:15}, \eqref{eq:17}, we obtain $\widehat{R}_{ik}^{\rm lb}[n] \triangleq \left\{  \begin{array}{l} 
	\widehat{r}^{\rm lb}_{ik}[n],\hfill \text{if} \; n \in {\cal T}_{ik}, \;\\
	0,\hfill \text{otherwise},
\end{array} \right.$, $\widehat{C}_{ik}^{\rm lb}[n] \triangleq
\delta_t \widehat{R}^{\rm lb}_{ik}[n], \; {\rm where} \; i \in \{1,2\}$. Moreover, we have $\widehat{R}_{ik}^{\rm lb}=\sum\limits_{n \in {\cal T}_{ik} } \widehat{R}_{ik}^{\rm lb}[n]$, $\widehat{C}_{ik}^{\rm lb} =\sum\limits_{n \in {\cal T}_{ik} } \widehat{C}_{ik}^{\rm lb}[n]$. Let us define ${\bf \Phi} \triangleq \{\Phi_{1k}^{\rm lb}[n], \Phi_{2k}^{\rm lb}[n], \forall k, n \}$ and ${\bf r}\triangleq\{\widehat{r}_{1k}^{\rm lb}[n],$ $\widehat{r}_{2k}^{\rm lb}[n], \forall k, n \}$. 

Bearing all the above developments in mind, we solve the following approximate convex program at the $(j+1)$-th iteration:
\begin{IEEEeqnarray}{rCl}
	\vspace{-0.01cm}
	\label{eq:P13}
	{\cal P}^{\rm FD}_{\rm convex}:\ &&\max_{\boldsymbol{\Psi} }~~ \sum\limits_{k \in {\cal K}} \lambda_k +  \mu \widehat{\mathbb{P}}^{(j)}(\boldsymbol \lambda) \IEEEyessubnumber\label{eq:P13:a}\\
	\mathtt{s.t.}~~ \vspace{-0.01cm}
	&&  \eqref{eq:P1:h}-   \eqref{eq:P1:m}, \eqref{eq:P11:b}, \eqref{eq:P12:c}, \eqref{eq:P12:d}, \eqref{eq:44}, 
	\IEEEyessubnumber\label{eq:P13:b}\\
	&&\lambda_k \frac{ S_k}{\widehat{R}_{1k}^{\rm lb}} \le (n_{{\rm end},k}-n_{{\rm start},k}+1)\delta_t, \forall k, \IEEEyessubnumber\label{eq:P13:c}\\ 
	\vspace{-0.01cm}
	&&\lambda_k \frac{ S_k}{\widehat{R}_{2k}^{\rm lb}} \le (N-n_{{\rm end},k})\delta_t, \forall k, \IEEEyessubnumber\label{eq:P13:d}\\
	\vspace{-0.01cm}
	&&\delta_t \; {{\rm min} \big({\widehat{R}_{1k}^{\rm lb}},{\widehat{R}_{2k}^{\rm lb}} \big)  \ge \lambda_k S_k, \forall k,} \IEEEyessubnumber\label{eq:P13:e}\\
	\vspace{-0.01cm}
	&& \sum\limits_{k=1}^K \delta_t {\widehat{R}_{2k}^{\rm lb}} \ge \sum\limits_{k=1}^K {\lambda_k S_k},  \IEEEyessubnumber\label{eq:P13:f} \\ \vspace{-0.01cm}
	&&\sum\limits_{k \in {\cal K}} \Big(\lambda_k S_k - \sum\limits_{l=n+1}^N  \delta_t {\widehat{R}_{1k}^{\rm lb}}[l] -  \sum\limits_{l=1}^{n-1} \delta_t {\widehat{R}_{2k}^{\rm lb}} [l] \Big) \notag\\ &&  \le C,   \forall k, n, \IEEEyessubnumber\label{eq:P13:g}
\end{IEEEeqnarray}
where $\boldsymbol{\Psi} \triangleq \{ {\bq}, {\bf a}, {\bf p}, {\boldsymbol \lambda}, {\bf z}, {\bf t}, {\bf\Phi}, {\bf r} \} $ and $\boldsymbol{\Psi}^{(j)} \triangleq \{ \bq^{(j)}, {\bf a}^{(j)}, {\bf p}^{(j)}, {\boldsymbol \lambda}^{(j)}, {\bf z}^{(j)},$  ${\bf t}^{(j)}, {\bf\Phi}^{(j)}\}$ as
the feasible point for \eqref{eq:P13} at iteration $j$. The convex program \eqref{eq:P13} can be solved by using standard convex optimization solvers \cite{Boy}. To ensure the feasibility of \eqref{eq:P13} at the first iteration, an appropriate starting point $\boldsymbol{\Psi}^{(0)}$ is necessary. This selection should be made such that the feasibility of \eqref{eq:P13:e}  is always guaranteed while additionally satisfying  other constraints. Therefore, we successively solve the following simplified version of \eqref{eq:P13}:
\begin{IEEEeqnarray}{rCl}
	\label{eq:P13_1}
	{\cal P}_{\rm feasible}^{\rm FD}:\ &&\max_{\boldsymbol{\Psi},\{\tau_k\}_{k=1}^K }~~ 	\vspace{-0.01cm}  \min_{\forall k}  {\tau_k} \IEEEyessubnumber\label{eq:P131:a}\\
	\mathtt{s.t.}~~ \vspace{-0.01cm}
	&&\delta_t \; {{\rm min} \big({\widehat{R}_{1k}^{\rm lb}},{\widehat{R}_{2k}^{\rm lb}} \big)  - \lambda_k S_k \ge \tau_k, \forall k,} \IEEEyessubnumber\label{eq:P131:b}\\
	&&  \eqref{eq:P13:b}-\eqref{eq:P13:d}, \eqref{eq:P13:f}, \eqref{eq:P13:g},
	\IEEEyessubnumber\label{eq:P131:d}
		\vspace{-0.01cm} 
\end{IEEEeqnarray} 
where $\tau_k$ is the slack variable. The initial feasible point $\Psi^{(0)}$ is obtained until  problem \eqref{eq:P13_1} is successfully solved and $\tau_k\geq 0,\forall k$. Then, the sub-optimal solution is obtained by successively solving \eqref{eq:P13} and updating the involved variables until satisfying the convergence condition (discussed below in detail). Finally, a pseudo-code for solving \eqref{eq:P1} is summarized in Algorithm \ref{Alg1}.
\begin{algorithm}[t]
	\begin{algorithmic}[1]
			\label{Alg1}
			\protect\caption{Proposed IA Based Design to Solve \eqref{eq:P1}}
			\label{alg_1}
			\global\long\def\algorithmicrequire{\textbf{Initialization:}}
			\REQUIRE  Set $j:=0$ and solve \eqref{eq:P13_1} to generate an initial feasible
			point $\boldsymbol{\Psi}^{(0)}$.
			\vspace{-0.01cm}
			\REPEAT
			\vspace{-0.01cm}
			\STATE Solve \eqref{eq:P13} to obtain the optimal solution $\boldsymbol{\Psi}^\star \triangleq \left(\bq^\star, {\bf a}^\star, {\bf p}^\star, {\boldsymbol \lambda}^\star, {\boldsymbol z}^\star, {\boldsymbol t}^{\star}, {\boldsymbol \Phi}^{\star}, {\boldsymbol r}^{\star}  \right)$.
			\vspace{-0.01cm}
			\STATE Update $\bq^{(j+1)}:=\bq^\star,{\bf a}^{(j+1)}:={\bf a}^\star, {\bf p}^{(j+1)} :={\bf p}^\star, {\boldsymbol \lambda}^{(j+1)} :={\boldsymbol \lambda}^\star, {\boldsymbol z}^{(j+1)}:={\boldsymbol z}^\star, {\boldsymbol t}^{(j+1)}:={\boldsymbol t}^{\star}$, ${\boldsymbol \Phi}^{(j+1)}:={\boldsymbol \Phi}^{\star}$.
			\vspace{-0.01cm}
			\STATE Set $j:=j+1.$
			\vspace{-0.01cm}
			\UNTIL Convergence \\
\end{algorithmic} \end{algorithm}
\subsection{Convergence and Complexity Analysis}
\vspace{-0.01cm}
\subsubsection{Convergence Analysis}
{Algorithm 1 is mainly based on inner approximation, where its convergence is
proved in \cite{marks,beck2010}. To be self-contained, we introduce the following proposition.}

{\begin{proposition}
	\label{proposition_1}
	The proposed Algorithm \ref{Alg1} yields a sequence of improved solutions converging to at least a local optimum of the relaxed problem ${\cal P}^{\rm FD}_{\rm relaxed}$.
\end{proposition}
\begin{proof}
	See Appendix~D.
\end{proof}}

\subsubsection{Complexity Analysis}
We now provide the worst-case complexity analysis for each iteration in Algorithm \ref{Alg1}. 
Since problem \eqref{eq:P13} is convex, several solvers employing the interior point method can be applied to solve efficiently \cite{Boy}. More specifically, the convex problem \eqref{eq:P13} involves $N(7+8K)+4K$ linear and quadratic constraints, and $5N(1+3K)+K$ scalar real variables. As a result, the per-iteration  computational complexity required to solve \eqref{eq:P13} is $\mathcal{O}(N(7+8K)+4K)^{0.5}(5N(1+3K)+K)^3$ \cite[Chapter 6]{ben2001lectures}. It results in the overall complexity of $\mathcal{O} \Big(N_i (N(7+8K)+4K)^{0.5}(5N(1+3K)+K)^3 \Big)$, where $N_i$ is the number of iterations to reach a local optimal solution.

\subsection{Throughput Maximization}
\label{sec:FDrate}
In an emergency case or during a natural disaster, data need to be collected timely to assess the current situation in a given area. The more collected information we have, the better our predictions are. This motivates us to present a new problem that maximizes the total amount of collected data with a given number of served IoT devices subjected to certain quality-of-service (QoS) constraints:
\begin{IEEEeqnarray}{rCl}\label{eq:P3}
	{\cal P}_{\rm rate}^{\rm FD}:\ &&\max_{\bq, {\bf a}, {\bf p}, {\boldsymbol \lambda} }~~ \sum_{k \in {\cal K}} \delta_t {\rm min} ( R_{1k},R_{2k} )  \IEEEyessubnumber \label{eq:P3:a}\\
	\mathtt{s.t.}~~ \vspace{-0.01cm}
	&& \left\|  {\boldsymbol \lambda}  \right\|_1  \ge \lambda_{\rm thresh}, \IEEEyessubnumber\label{eq:P3:b}\\
	&& \eqref{eq:P1:b}-\eqref{eq:P1:m},
	\IEEEyessubnumber\label{eq:P3:c}
\end{IEEEeqnarray}
where constraint \eqref{eq:P3:b} means that the total number of served IoT devices must be larger than or equal to a predefined threshold value, i.e., $\lambda_{\rm thresh}$. 

Similar to ${\cal P}^{\rm FD}$, ${\cal P}_{\rm rate}^{\rm FD}$ is also a mixed integer non-convex problem, which is NP-hard. Fortunately, by reusing the developments presented in Section III-B, \eqref{eq:P3} is rewritten as
\begin{IEEEeqnarray}{rCl}\label{eq:P2a}
 	{\cal P}_{\rm rate-convex}^{\rm FD}:\ &&\max_{\Psi }~~ \sum_{k \in {\cal K}} \delta_t {\rm min} (\widehat{R}_{1k}^{\rm lb} ,\widehat{R}_{2k}^{\rm lb} ) + \mu \widehat{\mathbb{P}}^{(j)}(\boldsymbol \lambda) \notag\\ \IEEEyessubnumber \label{eq:P4:a}\\
 	\mathtt{s.t.}~~ \vspace{-0.01cm}
 	&& \eqref{eq:P13:b}-\eqref{eq:P13:g}, \IEEEyessubnumber\label{eq:P4:b}
 \end{IEEEeqnarray}
where $\widehat{R}_{ik}^{\rm lb}$ are obtained as in Section \ref{sec:3b}.
\begin{algorithm}[t]
	\begin{algorithmic}[1]
		\label{Alg3}
		\protect\caption{Proposed IA-based Iterative Algorithm to Solve \eqref{eq:P3}}
		\label{alg_3}
		\global\long\def\algorithmicrequire{\textbf{Initialization:}}
		\REQUIRE  Set $j:=0$ and generate an initial feasible
		point $\boldsymbol{\Psi}^{(0)}$.
		\REPEAT
		\STATE Solve \eqref{eq:P2a} to obtain the optimal solution $\boldsymbol{\Psi}^\star \triangleq \left(\bq^\star, {\bf a}^\star, {\bf p}^\star, {\boldsymbol \lambda}^\star, {\boldsymbol z}^\star, {\boldsymbol t}^{\star}, {\boldsymbol \Phi}^{\star}, {\boldsymbol r}^{\star}  \right)$.
		\STATE Update $\bq^{(j+1)}:=\bq^\star,{\bf a}^{(j+1)}:={\bf a}^\star, {\bf p}^{(j+1)} :={\bf p}^\star, {\boldsymbol \lambda}^{(j+1)} :={\boldsymbol \lambda}^\star, {\boldsymbol z}^{(j+1)}:={\boldsymbol z}^\star, {\boldsymbol t}^{(j+1)}:={\boldsymbol t}^{\star},$ ${\boldsymbol \Phi}^{(j+1)}:={\boldsymbol \Phi}^{\star} $.
		\STATE Set $j:=j+1.$
		\UNTIL Convergence \\
\end{algorithmic} \end{algorithm}

Consequently, the solution of problem ${\cal P}_{\rm rate}^{\rm FD}$ can be found by successively solving a simpler convex problem in \eqref{eq:P2a}, as summarized in Algorithm \ref{Alg3}.

\section{Half Duplex Mode Scheme}
\label{sec:HD}
\subsection{Maximizing The Number of Served IoT Devices}
\label{sec:4a}
{In order to stress the benefits of our proposed method using FD mode, we now describe the problem again by considering HD mode at the UAV. First,} \eqref{eq:8} and \eqref{eq:9} can be rewritten as
\begin{align}
\label{eq:26}
&y_{ik}^{\rm{HD}} [n] = \sqrt{p_{ik}[n]} h_{ik}[n]x_{ik}[n] + n_0, \; i \in \{1,2\}.
\end{align}

In \eqref{eq:26}, the UAV only transmits data to GW when it finishes collecting data from all GUs in HD mode. Consequently, the RSI is disappeared compared to that of \eqref{eq:8}. Thus, the achievable rate (bits/s) of link from $k \to {\rm U}$ or ${\rm U} \to {\rm GW}$ to transmit the data of device $k$ at time slot $n$ is given as
\begin{align}
\label{eq:56}
r_{ik}^{\rm{HD}}[n] &= a_{ik}[n] B \log_2 \Bigg(1+\frac{ p_{ik}[n] | {\tilde h}_{1k}[n] |^2 \omega_0 }{  \big({H^2} + {{\left\| {\bq[n] - \bw} \right\|}^2}\big)^{\alpha/2} \sigma^2}\Bigg), \notag\\ & \qquad \qquad \qquad \qquad \qquad \qquad \qquad  \quad i \in \{1,2\},
\end{align}
where $\bw$ is $\bw_k$ and $\bw_0$ corresponding to $i$ equals 1 and 2, respectively.

Similar to \eqref{eq:Lemma1_2}, the approximated result of $r_{ik}^{\rm{HD}}[n]$ can be expressed as
{\begin{align}
	\label{eq:57}
	\bar{r}_{ik}^{\rm{HD}}[n] = a_{ik}[n] B \log_2 \Bigg(1+\frac{e^{-E} p_{ik}[n] \omega_0 }{  \big({H^2} + {{\left\| {\bq[n] - \bw} \right\|}^2}\big)^{\alpha/2} \sigma^2}\Bigg).
\end{align}}

By substituting \eqref{eq:57} into the equations \eqref{eq:15} and \eqref{eq:17}, we obtain $C_{1k}^{\rm{HD}}[n]$, $C_{2k}^{\rm{HD}}[n]=C_{2k}[n]$, and $R_{1k}^{\rm{HD}}[n]$, respectively. Then, we reformulate the problem  of maximizing the total number of served IoT devices as follows:
\begin{IEEEeqnarray}{rCl}\vspace{-0.01cm} \label{eq:P2}
	{\cal P}^{\rm HD}:\ &&\max_{\bq, {\bf a}, {\bf p},{\boldsymbol \lambda} }~~ \left\|  {\boldsymbol \lambda}  \right\|_1 \IEEEyessubnumber\label{eq:P2:a} \\
	\mathtt{s.t.}~~ \vspace{-0.01cm}
	&&\eqref{eq:P1:b},   \eqref{eq:P1:f}, \eqref{eq:P1:h}-\eqref{eq:P1:m}, \IEEEyessubnumber\label{eq:P2:b} \\ \vspace{-0.01cm}
	&&\delta_t \; {{\rm min} (R_{1k}^{\rm{HD}},R_{2k}^{\rm HD})  \ge \lambda_k S_k, \forall k,} \IEEEyessubnumber\label{eq:P2:c}
	\\ \vspace{-0.01cm}
	&&\lambda_k \frac{ S_k}{R_{1k}^{\rm HD}} \le (n_{{\rm end},k}-n_{{\rm start},k}+1)\delta_t, \forall k, \IEEEyessubnumber\label{eq:P2:d}\\  \vspace{-0.01cm}
	&&\sum\limits_{k \in {\cal K}} \Big(\lambda_k S_k -  \sum\limits_{l=n+1}^N \delta_t R_{1k}^{\rm HD}[l] -   \sum\limits_{l=1}^{n-1} \delta_t R_{2k}^{\rm HD}[l] \Big) \notag\\ &&\le C, \forall n. \IEEEyessubnumber\label{eq:P2:e} 
\end{IEEEeqnarray} 

The problem ${\cal P}^{\rm HD}$ is a mixed integer non-convex due to the binary constraint \eqref{eq:P1:b} and non-convex constraints \eqref{eq:P1:f}, \eqref{eq:P2:c}, \eqref{eq:P2:d}, and \eqref{eq:P2:e}. In order to seek a suitable solution, we first relax binary constraint \eqref{eq:P1:b} as in \eqref{eq:P12:b}. Then, by introducing $z_{1k}^{\rm{HD}}[n]$ and $z_{2k}^{\rm{HD}}[n]$ such that $\left({H^2} + {\left\| {\bq[n] - \bw_k} \right\|}^2\right) \le (z_{1k}^{\rm{HD}}[n])^{2/\alpha}$ and $\left({H^2} + {\left\| {\bq[n] - \bw_0} \right\|}^2\right) \le (z_{2k}^{\rm{HD}}[n])^{2/\alpha}$, \eqref{eq:57} can be expressed as
{\begin{align}
\label{eq:58}
\bar{r}_{ik}^{\rm{HD}}[n] &= a_{ik}[n] B  \log_2 \left(1+\frac{e^{-E} p_{ik}[n] \omega_0 }{  z_{ik}^{\rm{HD}}[n] \sigma^2}\right), \; {\rm with} \; i \in \{1,2\}.
\end{align}}

{Given that the $\bar{r}_{ik}^{\rm{HD}}[n]$ is the same as $\bar{r}_{2k}^{\rm lb}[n]$ in \eqref{eq:Lemma1_2}, we apply IA method for $\bar{r}_{2k}^{\rm lb}[n]$ in Section~\ref{Sec:3} to $\bar{r}_{ik}^{\rm{HD}}[n]$.} As a result, $r_{ik}^{\rm{HD}}[n]$ can be rewritten as
\vspace{-0.01cm}
\begin{align}
\label{eq:60}
\bar{r}_{ik}^{\rm{HD}}[n]= a_{ik}[n]  \Phi_{ik}^{\rm{HD}}[n],
\end{align}
\vspace{-0.01cm}
where
{\begin{align}
\label{eq:63}
\Phi_{ik}^{\rm{HD}}[n] &= B 
\log_2 \left(1+\frac{e^{-E} p_{ik}[n] \omega_0 }{z_{ik}^{\rm{HD}}[n]\sigma^2}\right). 
\end{align} }
Similar to \eqref{eq:362}, $\Phi_{ik}[n]$ is lower bounded by
\begin{align}
\label{eq:64}
\Phi_{ik}^{\rm{HD}}[n] \ge \bar{\Phi}_{ik}^{\rm{HD}}[n],
\end{align}
where $\bar{\Phi}_{1k}^{\rm{HD}}[n]$ and $\bar{\Phi}_{2k}^{\rm{HD}}[n]$ can be calculated as $\bar{\Phi}_{2k}[n]$, shown in Appendix~B.

As in \eqref{eq:38}, it follows that
\begin{align} \label{eq:66}
r_{ik}^{\rm HD}[n]  \ge r_{ik}^{\rm HD,lb}[n] = a_{ik}[n]  \Phi_{ik}^{\rm HD,lb}[n], 
\end{align}  
where $\Phi_{ik}^{\rm HD,lb}[n]$ is a slack variable which is a lower bound of $\bar{\Phi}_{ik}^{\rm{HD}}[n]$. Then, by applying the first order Taylor approximation for $a_{ik}[n]  \Phi_{ik}^{\rm HD,lb}[n]$, it yields:
\begin{align} \label{eq:68}
r_{ik}^{\rm HD,lb}[n]  \ge \bar{r}_{ik}^{\rm HD,lb}[n],
\end{align} 
where $\bar{r}_{ik}^{\rm HD,lb}[n]$ and $\bar{r}_{ik}^{\rm HD,lb}[n]$ can be represented as in \eqref{eq:36}. 

In turn, by introducing a slack variable $\widehat{r}_{ik}^{\rm HD,lb}[n]$, constraint \eqref{eq:68} is innerly approximated by the following convex constraints:
\begin{align} \label{eq:70}
\bar{r}_{ik}^{\rm HD,lb}[n]  \ge \widehat{r}_{ik}^{\rm HD,lb}[n].
\end{align}

By substituting $\widehat{r}_{ik}^{\rm HD,lb}[n]$ into \eqref{eq:17}, we obtain $\widehat{R}_{ik}^{\rm HD,lb}[n]$. Moreover, we have $\widehat{R}_{ik}^{\rm HD,lb}=\sum\limits_{n \in {\cal T}_{ik} } \widehat{R}_{ik}^{\rm HD,lb}[n]$. In Algorithm \ref{Alg2}, we propose an iterative algorithm to solve the problem \eqref{eq:P2}. At the $(j+1)$-th iteration, it solves the following convex program:
\begin{IEEEeqnarray}{rCl}\label{eq:P21}
	{\cal P}_{\rm convex}^{\rm HD}:\ &&\max_{\Psi }~~ \sum\limits_{k \in {\cal K}} \lambda_k + \mu \widehat{\mathbb{P}}^{(j)}(\boldsymbol \lambda)  \IEEEyessubnumber\label{eq:P21:a} \\
	\mathtt{s.t.}~~\vspace{-0.01cm}
	&&  \eqref{eq:P1:h}-\eqref{eq:P1:m},\eqref{eq:P11:b}, \eqref{eq:70}, \IEEEyessubnumber\label{eq:P21:b} \\
	&&\delta_t \; {{\rm min} (\widehat{R}_{1k}^{\rm{HD,lb}},\widehat{R}^{\rm{HD,lb}}_{2k})  \ge \lambda_k S_k, \forall k,} \IEEEyessubnumber\label{eq:P21:c}
	\\ 	\vspace{-0.01cm} 
	&& \lambda_k \frac{ S_k}{\widehat{R}_{1k}^{\rm{HD,lb}}} \le (n_{{\rm end},k}-n_{{\rm start},k}+1)\delta_t, \forall k, \IEEEyessubnumber\label{eq:P21:d}\\
	&&\sum\limits_{k \in {\cal K}} \Big(\lambda_k S_k -  \sum\limits_{l=n+1}^N \delta_t \widehat{R}_{1k}^{\rm{HD,lb}}[l] -   \sum\limits_{l=1}^{n-1} \delta_t \widehat{R}_{2k}^{\rm{HD,lb}}[l] \Big) \notag\\&&  \le C,  \forall n, \IEEEyessubnumber\label{eq:P21:e} \\&& \lambda_k \frac{ S_k}{\widehat{R}_{2k}^{\rm{HD,lb}}} \le (N-n_{{\rm end},k})\delta_t, \forall k, \IEEEyessubnumber\label{eq:P21:f}\\
	&&\left({H^2} + {\left\| {\bq[n] - \bw_k} \right\|}^2\right) \le (z_{1k}^{\rm{HD}}[n])^{2/\alpha}, \notag\\
	&& \left({H^2} + {\left\| {\bq[n] - \bw_0} \right\|}^2\right) \le (z_{2k}^{\rm{HD}}[n])^{2/\alpha}. \IEEEyessubnumber\label{eq:P21:g}
\end{IEEEeqnarray} 

Similar to \eqref{eq:P11}, we adopt a penalty function in objective to guarantee an exact binary value of $\lambda_k$, $\forall k \in \mathcal{K}$. The initial feasible point to solve \eqref{eq:P21} can be obtained similar to \eqref{eq:P13_1}. 

\begin{algorithm}[t]
	\begin{algorithmic}[1]
		\label{Alg2}
		\protect\caption{Proposed IA-based Iterative Algorithm to Solve \eqref{eq:P2}}
		\global\long\def\algorithmicrequire{\textbf{Initialization:}}
		\REQUIRE  Set $j:=0$ and generate an initial feasible
		point ${\boldsymbol \Psi}^{(0)}$.
		\vspace{-0.01cm}
		\REPEAT
		\vspace{-0.01cm}
		\STATE Solve \eqref{eq:P21} to obtain the optimal solution ${\boldsymbol \Psi}^\star \triangleq \left(\bq^\star, {\bf a}^\star, {\bf p}^\star, {\boldsymbol \lambda}^\star, {\boldsymbol z}^\star, {\boldsymbol \Phi}^{\star}, {\boldsymbol r}^{\star} \right)$.
		\vspace{-0.01cm}
		\STATE Update $\bq^{(j+1)}:=\bq^\star,{\bf a}^{(j+1)}:={\bf a}^\star, {\bf p}^{(j+1)} :={\bf p}^\star, {\boldsymbol \lambda}^{(j+1)} :={\boldsymbol \lambda}^\star, {\boldsymbol z}^{(j+1)}:={\boldsymbol z}^\star,$ ${\boldsymbol \Phi}^{(j+1)}:={\boldsymbol \Phi}^\star$.
		\UNTIL Convergence\\
\end{algorithmic} \end{algorithm}

\begin{algorithm}[t]
	\begin{algorithmic}[1]
		\label{Alg4}
		\protect\caption{Proposed IA-based Iterative Algorithm to Solve \eqref{eq:P50}}
		\label{alg_5}
		\global\long\def\algorithmicrequire{\textbf{Initialization:}}
		\REQUIRE  Set $j:=0$ and generate an initial feasible
		point ${\boldsymbol \Psi}^{(0)}$.
		\REPEAT
		\STATE Solve \eqref{eq:P5} to obtain the optimal solution ${\boldsymbol \Psi}^\star \triangleq \left(\bq^\star, {\bf a}^\star, {\bf p}^\star, {\boldsymbol \lambda}^\star, {\boldsymbol z}^\star,  {\boldsymbol \Phi}^{\star}, {\boldsymbol r}^{\star}  \right)$.
		\STATE Update $\bq^{(j+1)}:=\bq^\star,{\bf a}^{(j+1)}:={\bf a}^\star, {\bf p}^{(j+1)} :={\bf p}^\star, {\boldsymbol \lambda}^{(j+1)} :={\boldsymbol \lambda}^\star, {\boldsymbol z}^{(j+1)}:={\boldsymbol z}^\star$, ${\boldsymbol \Phi}^{(j+1)}:={\boldsymbol \Phi}^\star$.\\
		\STATE Set $j:=j+1.$
		\UNTIL Convergence \\
\end{algorithmic} \end{algorithm}

\subsubsection{Complexity Analysis}:
The convex problem \eqref{eq:P21} involves $N(7+8K)+4K$ linear and quadratic constraints, and $3N(1+4K)+K$ scalar real variables.  As a result, the per-iteration complexity required to solve \eqref{eq:P21} is $(N(7+8K)+4K)^{0.5}(3N(1+4K)+K)^3$. It results in the overall complexity is $\mathcal{O} \Big(N_i (N(7+8K)+4K)^{0.5}(3N(1+4K)+K)^3\Big) $, with $N_i$ is the number of iterations to reach a local solution.

\subsection{Throughput Maximization}
\label{sec:HDrate}
In this section, we reuse all the slack variables as introduced in Sections \ref{sec:FDrate} and \ref{sec:4a}. First, the throughput maximization problem for HD mode can be presented as:
\begin{IEEEeqnarray}{rCl}\label{eq:P50}
	{\cal P}_{\rm rate}^{\rm HD}:\ &&\max_{\bq, {\bf a}, {\bf p}, {\boldsymbol \lambda} }~~ \sum_{k \in {\cal K}} \delta_t {\rm min} ( R_{1k}^{\rm HD},R_{2k}^{\rm HD} ) \notag\\ \IEEEyessubnumber \label{eq:P50:a}\\
	\mathtt{s.t.}~~ \vspace{-0.01cm}
	&& \eqref{eq:P3:b}, \eqref{eq:P2:b}-\eqref{eq:P2:e}. \IEEEyessubnumber\label{eq:P50:b}
\end{IEEEeqnarray}

By following the same steps presented in Section \ref{sec:FDrate}, we obtain the following convex optimization problem:
\vspace{-0.01cm}
\begin{IEEEeqnarray}{rCl}\label{eq:P5}
	{\cal P}_{\rm rate-convex}^{\rm HD}:\ && \max_{\boldsymbol{\Psi}}~~ \sum_{k \in {\cal K}} \delta_t {\rm min} (\widehat{R}_{1k}^{\rm{HD,lb}},\widehat{R}_{2k}^{\rm{HD,lb}}) + \mu \widehat{\mathbb{P}}^{(j)}(\boldsymbol \lambda) \IEEEyessubnumber \label{eq:P5:a}\\
	\mathtt{s.t.}~~
	&& \eqref{eq:P3:b}, \eqref{eq:P21:b}- \eqref{eq:P21:g},  \IEEEyessubnumber\label{eq:P5:b}
\end{IEEEeqnarray}
where ${\widehat{R}_{ik}^{\rm{HD,lb}}}$ can be obtained as in Section IV-A. Due to the convexity of problem ${\cal P}_{\rm rate}^{\rm HD}$, the solution of problem ${\cal P}_{\rm rate}^{\rm HD}$ can be iteratively obtained as in Algorithm \ref{Alg4}.

\section{Numerical Results}
\label{Sec:Num}
In this section, we present numerical results to evaluate the proposed joint bandwidth allocation and transmit power for the devices/UAV as well as the UAV trajectory design in UAV-assisted IoT networks. We consider a system with $K$ IoT devices that are randomly  distributed in a horizontal plane, i.e, ${\rm area}=x^2$ $(m^2)$, with {$x=$ 500 m.} We assume that the GW, the initial location, and end location of the UAV are located at {(0, 500 m), $\bq_{\rm I}$ = [500 m, 200 m], and $\bq_{\rm F}$ = [300 m, 0],} respectively. The UAV flight altitude is invariant at $H=100$ m \cite{Y_Zeng_1}. {The total bandwidth is $B=20$ MHz. Thus, the total AWGN power is $\sigma^2=-174+10\log_{10}(B)=-100.9897$ dBm. The transmit power budget of the UAV and IoT devices is respectively set as $P_{\rm U}^{\rm max}=$ 18 dBm and $P_k^{\rm max}=$ 10 dBm.} Other parameters are set as follows: maximum speed $V_{\max}=50$ m/s, {path loss exponent $\alpha=2.4$,} $\omega_0=$ -30 dB, $S_k \in$ [10, 70] Mbits, one time slot duration $\delta_t=0.5$ s, the maximum collection time deadline for each device $k$ $n_{{\rm end},k}$ is uniformly distributed between $n_{{\rm end}, k}^{ \min }$ and $n_{{\rm end}, k}^{ \max }$. {The RSI suppression $\rho^{\rm RSI}$ is set to -80 dB \cite{bharadia2013full,zhang2018wideband}}. To show the superiority of our designs, we compare the proposed methods with benchmark schemes. Herein, the benchmark FD 2 (BFD2) and benchmark HD 2 (BHD2) are respectively implemented similar to Algorithms \ref{Alg1} and \ref{Alg3} with fixed resource allocation, i.e., $a_{1k}[n]=a_{1k}[n]=\frac{1}{K}$, $p_{1k}[n]=P_k^{\max}[n]$, $p_{2k}[n]=\frac{P_U^{\max}}{K}$. The benchmark FD 1 (BFD1) and benchmark HD 1 (BHD1) are implemented with a fixed trajectory, i.e., linear from initial to final locations.
\begin{figure*}[t]
	\centering    
	\subfigure[FD mode] {\label{fig:3a}\includegraphics[width=9cm,height=7cm]{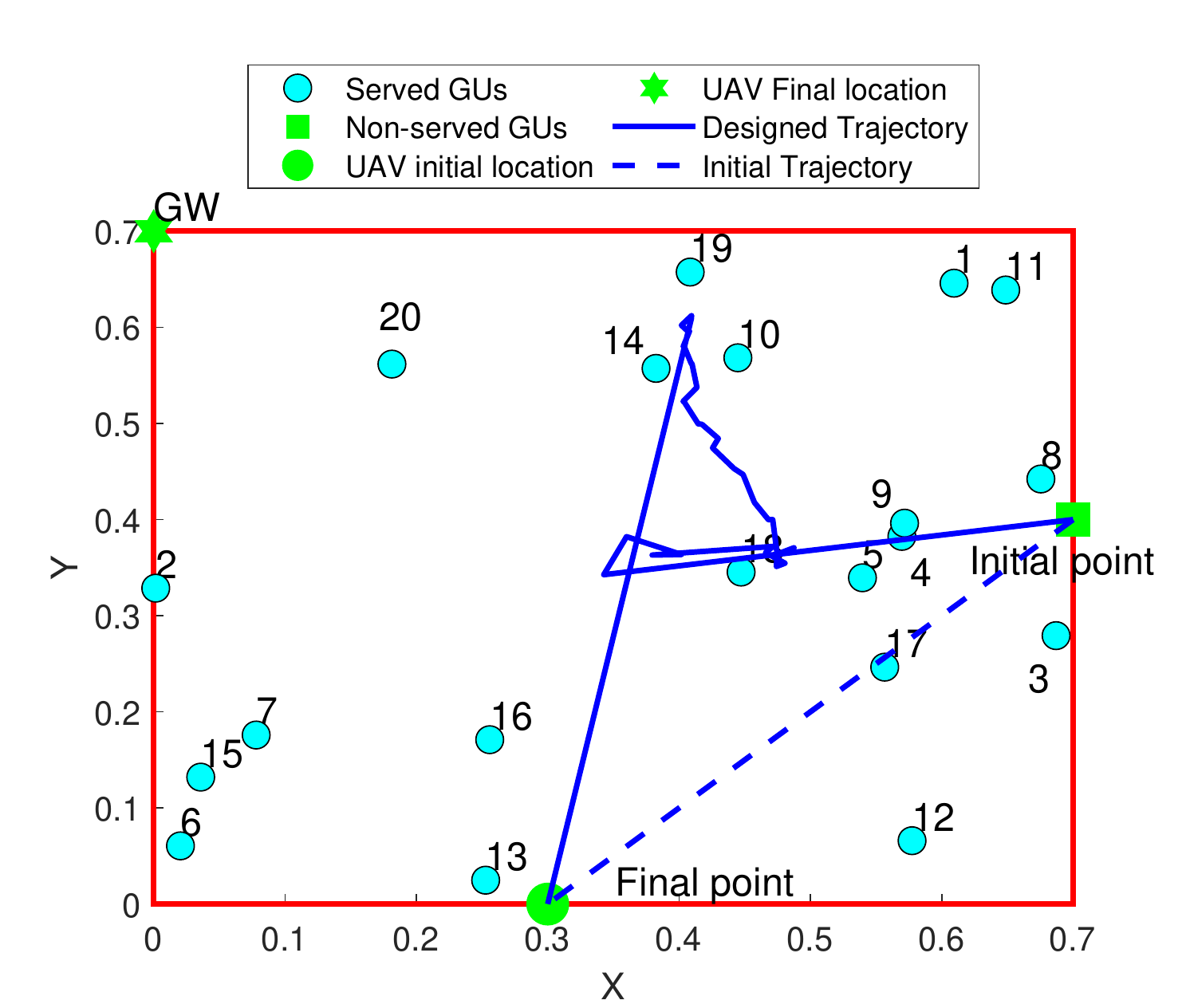}}
	\subfigure[HD mode] {\label{fig:3b}\includegraphics[width=9cm,height=7cm]{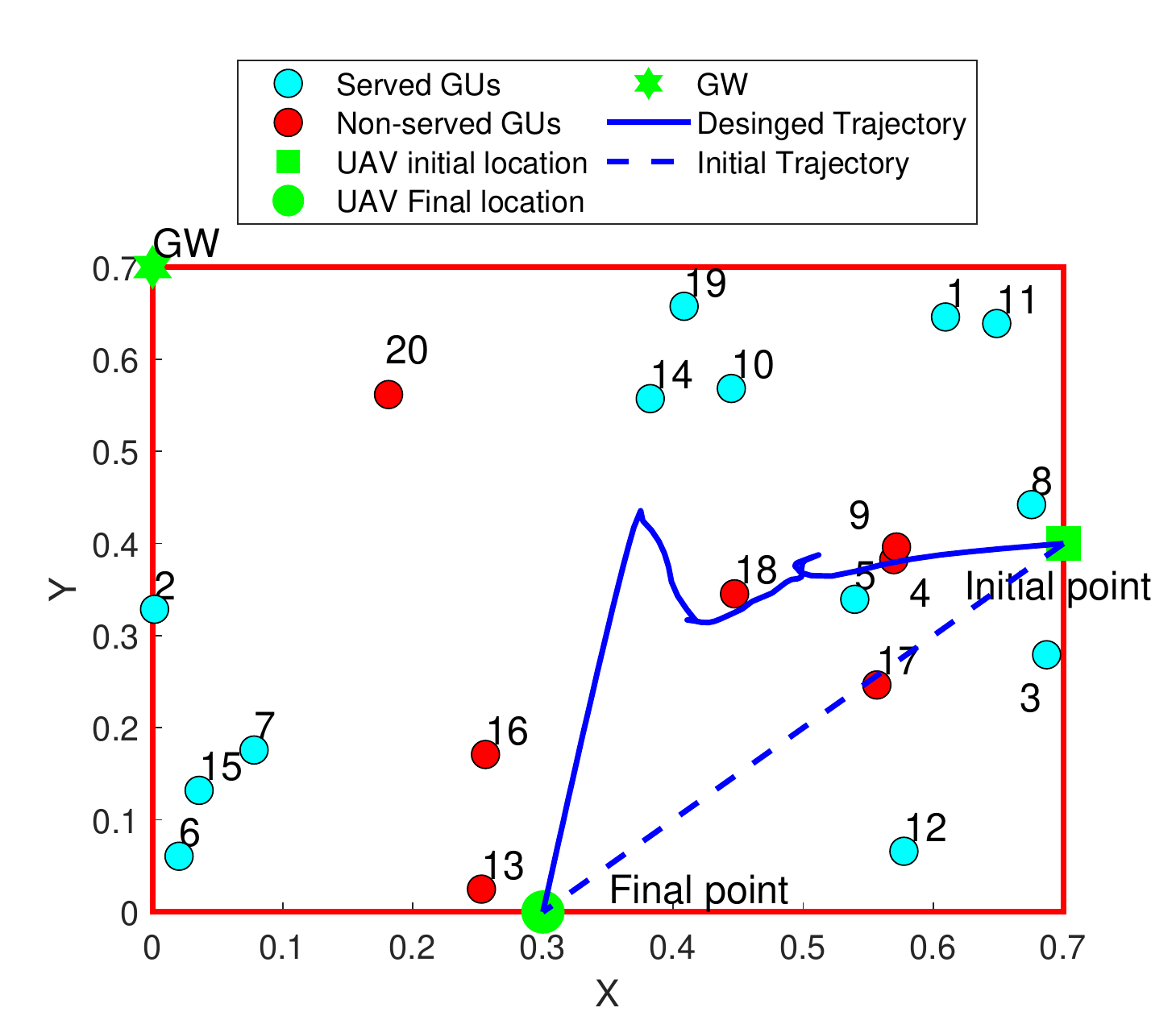}}
	\caption{Geometry distribution of GUs and the UAV trajectory.}
	\label{fig:3}
\end{figure*}

\begin{figure}[t]
	\centering
	\includegraphics[width=9cm,height=7cm]{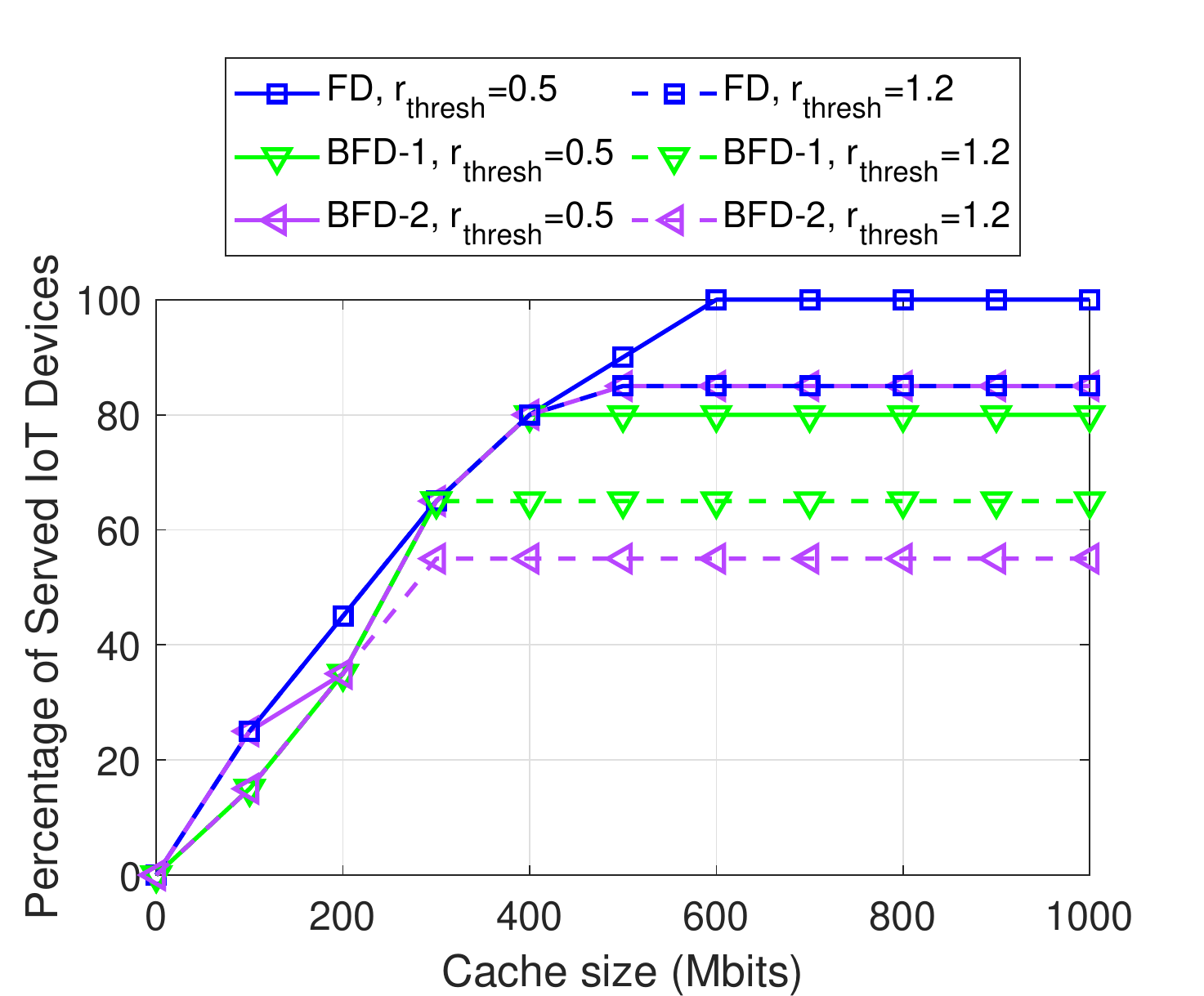}
	\caption{Percentage of served IoT devices vs. cache size in FD mode with different value of $r_{\rm thresh}$.}
	\label{fig:4}
\end{figure}		
\begin{figure}[t]
		\centering
		\includegraphics[width=9cm,height=7cm]{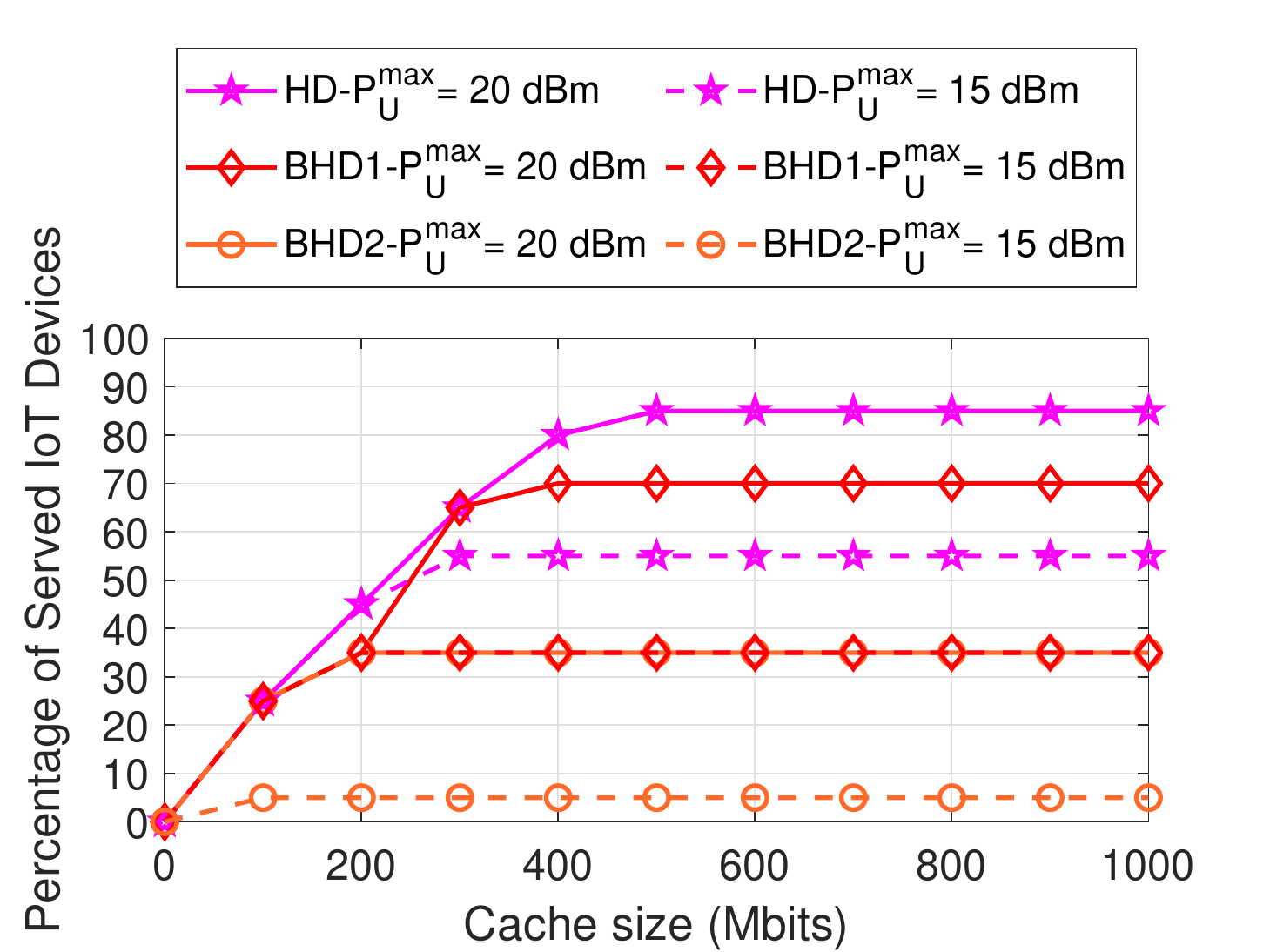}
		\caption{Percentage of served IoT devices vs. cache size in HD mode with different value of $P_k^{\max}[n]$. }
		\label{fig:5}
\end{figure}

\begin{figure*}[t]
	\centering   
	\subfigure[Full-duplex mode] {\label{fig:6a}\includegraphics[width=9cm,height=7.5cm]{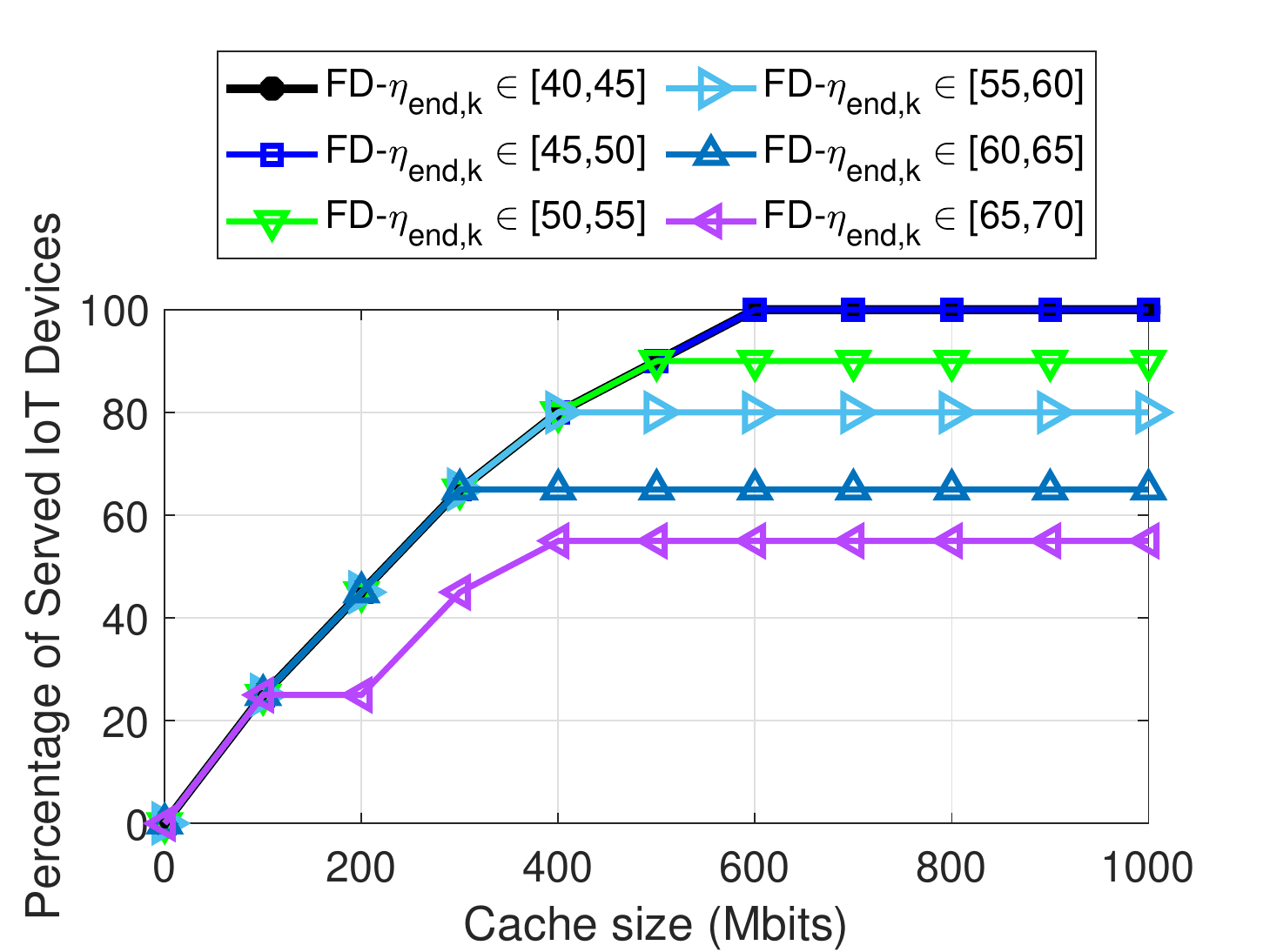}}
	\subfigure[Half-duplex mode] {\label{fig:6b}\includegraphics[width=9cm,height=7.5cm]{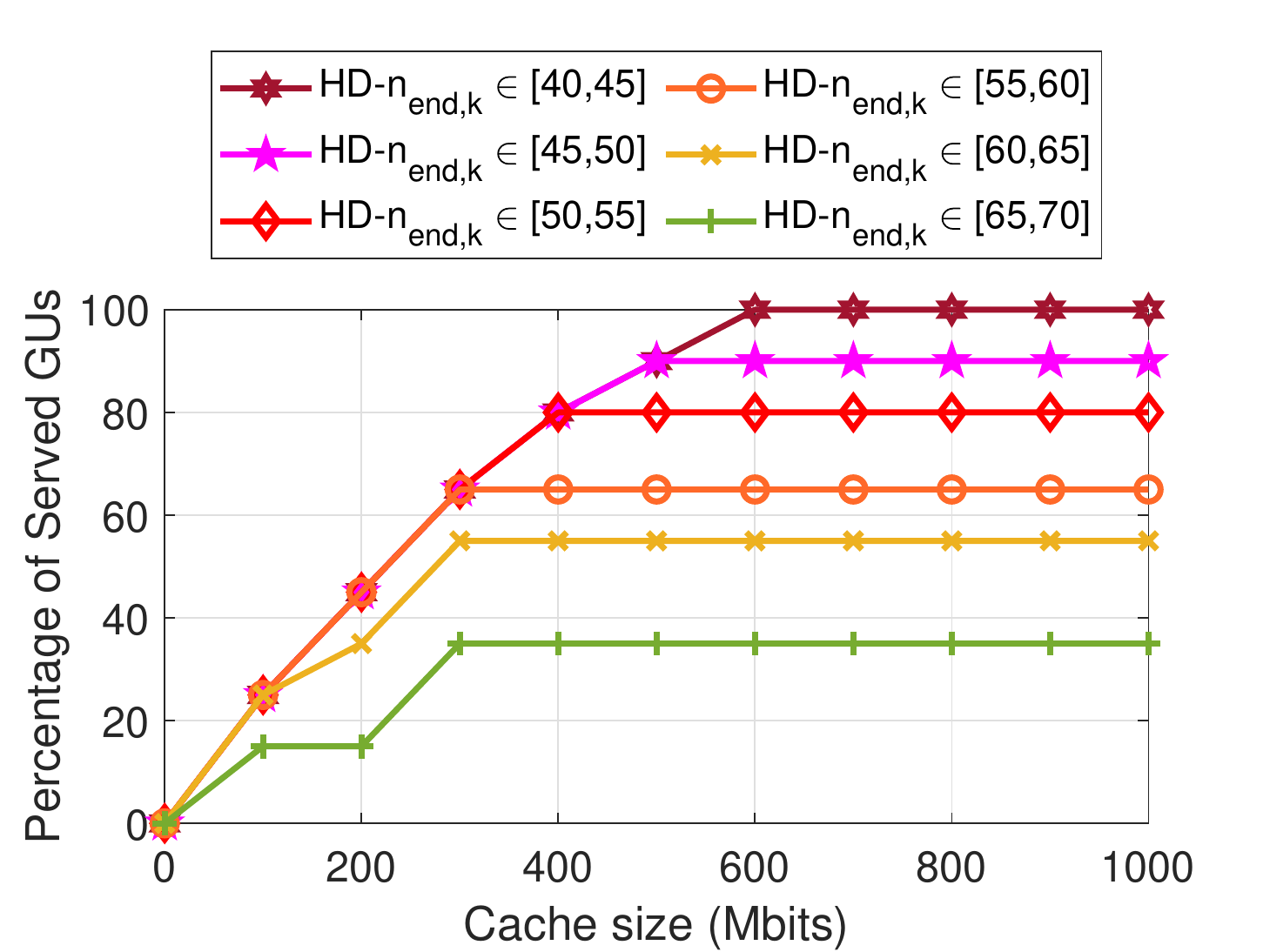}}
	\caption{Percentage of served IoT devices vs. cache size with different range of $\eta_{\rm end,k}$.}
	\label{fig:6}
\end{figure*}

\begin{figure}
		\centering
		\includegraphics[width=9cm,height=7cm]{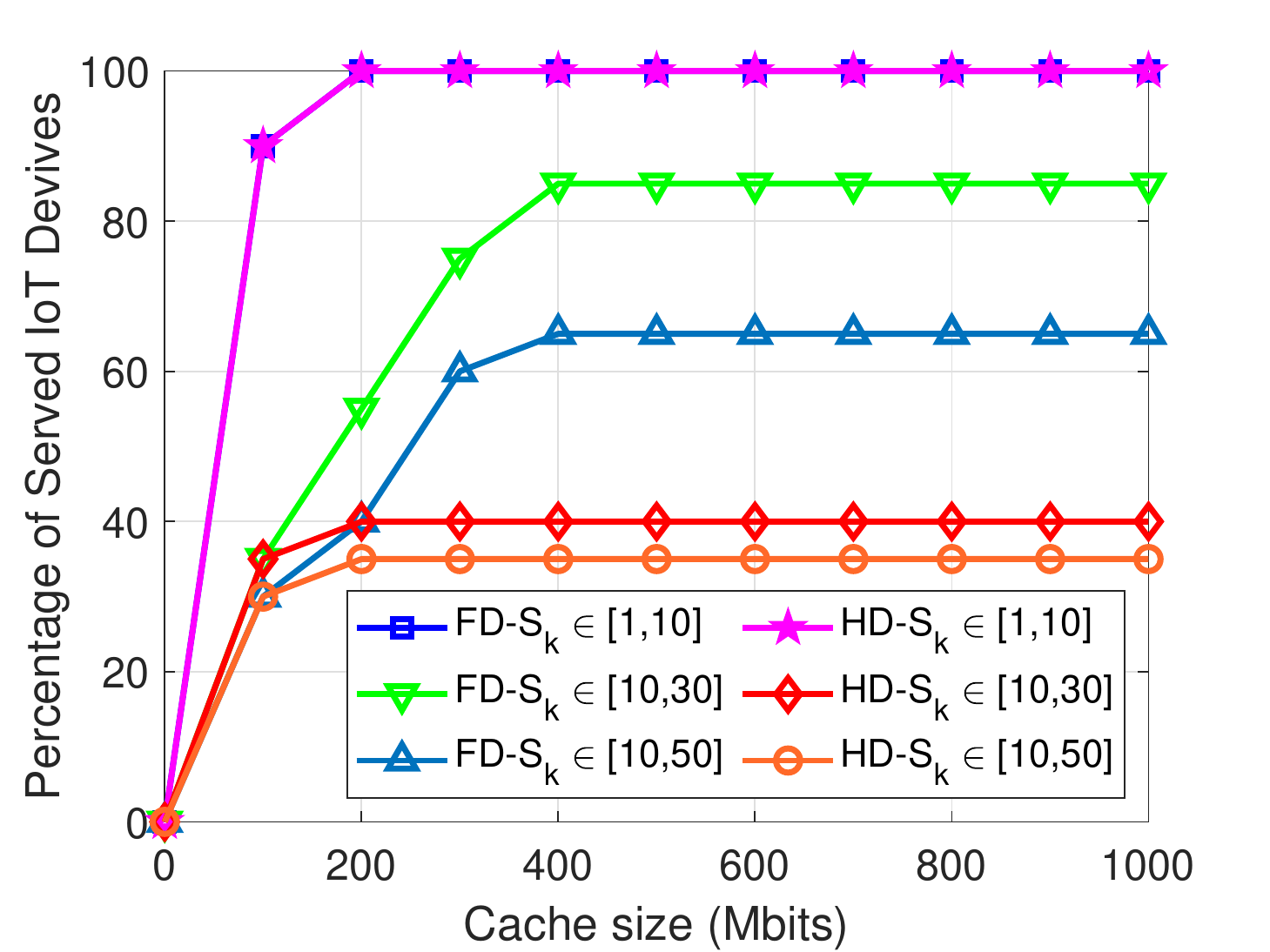}
		\caption{Percentage of served IoT devices vs. cache sizes with different $S_k$.}
		\label{fig:7}
\end{figure}
\begin{figure}
		\centering
		\includegraphics[width=9cm,height=7cm]{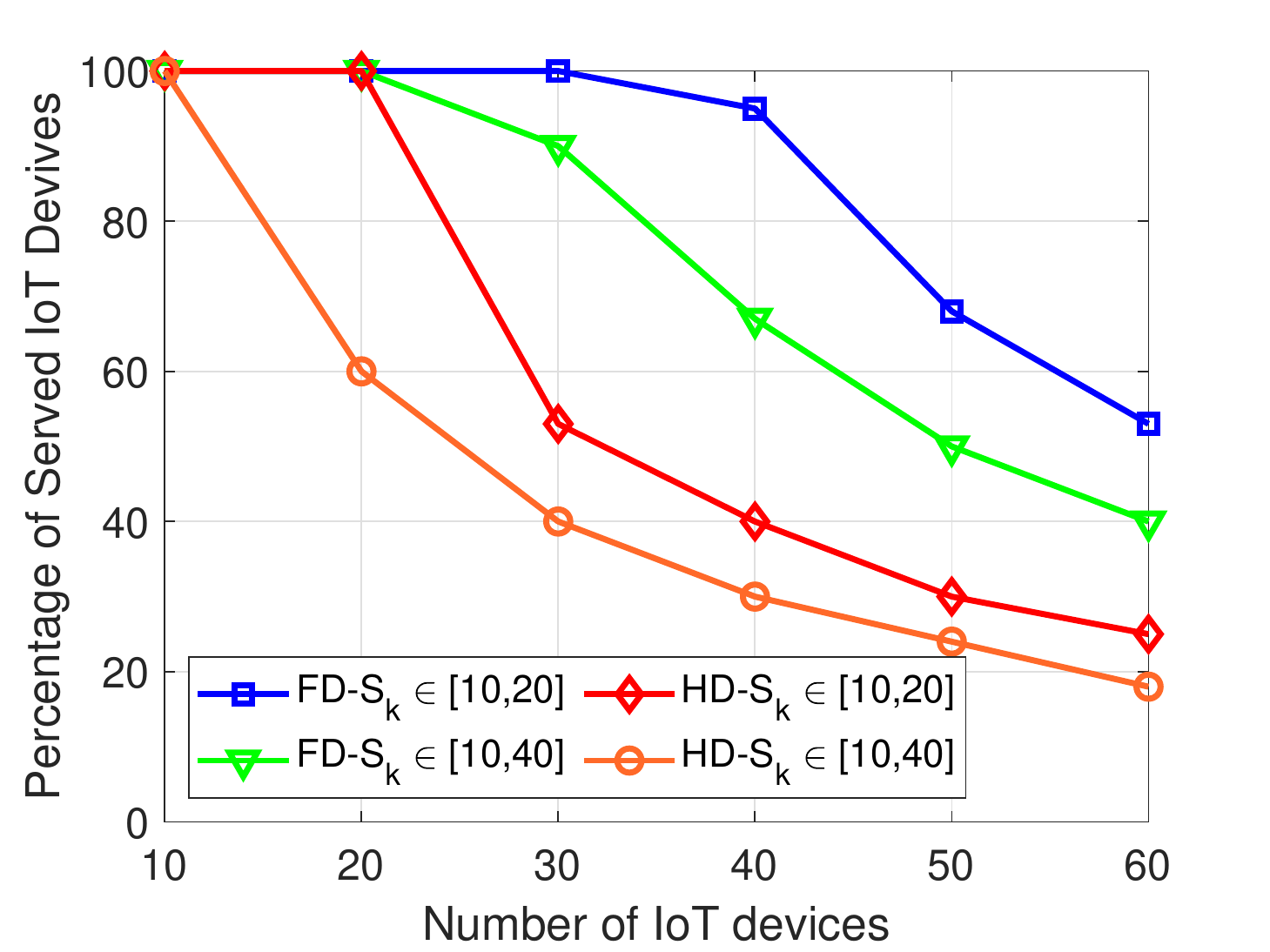}
		\caption{Percentage of served IoT devices vs. network size (maximum IoT devices located in the network area). }
		\label{fig:8}
\end{figure}

\subsection{Maximizing the Number of Served IoT Devices}
{Fig. \ref{fig:3}  plots the UAV's designed trajectory corresponding to FD and HD mode, with $N=$ 70 times slots, $\eta_{{\rm start},k}\in [2,15]$, $\eta_{{\rm end},k}\in [25,50]$, ${\rm area}=$ {700 m $\times$ 700 m}, $C=1000$, and $S_k$ values are ranging from 10 to 55 Mbits, $P_{\rm U}^{\rm max}=$ 19 dBm and $P_k^{\rm max}=$ 10 dBm. In additions, the GW, initial location, and end location of the UAV are respectively set as $(0,700 \;  \rm m)$, ${\bq_{\rm I}}=(700 \; \rm m, 400 \;  \rm m)$, ${\bq_{\rm I}}=(300 \;  \rm m, 0 )$.} First, we observe that the proposed FD method significantly improves the number of served IoT devices than the HD method, {i.e., 20 and 13 served GUs in FD and HD mode, respectively. Besides, the UAV can fly closer to GW and GUs in FD than in HD mode.} It is because the UAV transfers device $k$'s data to GW right after it finishes gathering data of that IoT device in FD-based scheme. While in HD mode, the UAV only operates in the downlink transmission when it completes the data acquisition for all users on the uplink to prevent RSI at the UAV. {Consequently, the UAV in the FD scheme has more time to fly closer to GW and GUs.} Thus, it obtains a higher probability of satisfying the GUs' RT. The UAV in the HD mode can collect information and fulfill the latency constraint for each IoT device, {but it has less time to move forward GUs/GW to collect/offload the generated data. Thus, the performance in the HD-based method is degraded.}

{In Fig. \ref{fig:4}, we investigate the performance of FD-based schemes with different QoS requirements. Specifically, the QoS is defined as the minimum rate threshold at the UAV/GW to successfully decode the signal, i.e., $r_{1k,{\rm thresh}}[n]$ and $r_{2k,{\rm thresh}}[n]$. For simplicity, we assume that $r_{1k,{\rm thresh}}[n] = r_{2k,{\rm thresh}}[n] = r_{\rm thresh}$. It can be seen that the more the minimum rate threshold is required, the fewer users the system can serve. This is because the UAV tends to come closer or spend more time around an IoT device to gain a higher rate requirement. As a result, the UAV has less chance of serving more devices due to limited flight time and latency constraints per IoT user. Another observation is that for larger cache sizes, the number of served users increases. It is due to the fact that the UAV has more capacity to store incoming data. Thus, the UAV can serve more users before offloading information to GW. Similar to Fig. \ref{fig:3}, our proposed FD algorithm achieves a much better percentage of served IoT devices compared to BFD1 and BFD2 schemes, respectively. Particularly, the performance of the BFD2 outperforms BFD1 with a small QoS requirement, i.e., $r_{\rm thresh}$ = 0.5. However, the BFD2's performance is inferior to that of BFD1 method with a large QoS value, i.e., $r_{\rm thresh}$ = 1.2. This is due to the fixed resource allocation per each time slot $n$ in these algorithms. This additionally leads to fluctuations in data transmission rate values with low variance during time slot $n$, i.e., $r_{1k}[n]$ and $r_{2k}[n]$. Thus, when the $r_{\rm thresh}$ value is still lower than the average rate of the BFD2, the performance is not significantly affected. Nevertheless, if $r_{\rm thresh}$ is large enough, the performance of BFD2 will drastically be influenced.}

{Fig. \ref{fig:5} depicts the percentage of served IoT devices versus cache size with different value of $P_k^{\max}[n]$. The parameters are set up similarly as shown in Fig. \ref{fig:4}, e.g., $r_{\rm thresh} = 0.5$. First, we observe that HD-based schemes' performance is interior to that of FD counterparts. In particular, at $P_{\rm U}^{\max}=20$ dBm and $C=800$, the HD method only serves up to 85 $\%$ number of users, while the FD scheme can serve all IoT devices with $P_{\rm U}^{\max}=18$ dBm and $C=800$, as shown in Fig. \ref{fig:4}. This also confirms the advantages of the FD system. Second, it can be easily seen that the HD scheme outperforms benchmark ones, i.e., BHD1 and BHD2. Specifically, at $P_{\rm U}^{\max}=20$ dBm and $C=500$, the HD algorithm can serve 85$\%$ of GUs, and the BHD1 achieves less than 15$\%$ OP. In comparison, the BHD2 scheme imposes a 35$\%$ percentage of served IoT devices. In Figs. \ref{fig:4} and \ref{fig:5}, the proposed FD and HD algorithms provide significantly better performance than those benchmarks, which shows the superiority of these designed schemes compared to other ones.}

{Fig. \ref{fig:6} shows the impact of different value of $\eta_{\rm end,k}$ on our system, with $N=$ 80, $K$ = 20, area = 500 m $\times$ 500 m, $P_{\rm U}^{\rm max}=$ 18 dBm, $P_k^{\rm max}=$ 15 dBm, $\eta_{{\rm start},k} \in [2,20]$, and $S_k$ value is ranging from 10 to 55 Mbits. It is observed that the percentage of served users increases corresponding to $\eta_{{\rm end},k} \in$ $[65,70],$ $[60,65],$ $[55,60]$, $[50,55]$, $[45,50],$ $[40,45],$ respectively. It can be explained by constraint \eqref{eq:P1:c}, which describes the condition of the user being successfully served. Since the total throughput collected is proportional to the time duration allocated to the UL/DL. When the given time for UL from an IoT device to a UAV is large enough, the number of served IoT users depends significantly on the time allocation for DL from UAV to GW.} Furthermore, the time period for DL is calculated as $N-\eta_{{\rm end},k}^{\rm min}$ and $N-\eta_{{\rm end},k}^{\rm max}$ for the FD and HD schemes, respectively. We see that the period of time allocated for DL in the FD algorithm is higher than that in the HD algorithm, such that the performance of the FD scheme outperforms the HD one. Specifically, the total number of served IoT users obtained from the HD scheme equals that of the FD method when the value of $N-\eta_{{\rm end},k}^{\rm max}$ is large enough. {For instance, in Figs. \ref{fig:6a} and \ref{fig:6b}, both proposed methods can serve the maximum number of IoT devices when $\eta_{{\rm end},k} \in [40,45]$ and $C \ge 600$. In this scenario, the UAV should work in HD mode for simplicity of operation in realistic implementation.}

\begin{figure}[t]
	\centering      
	\includegraphics[width=9cm,height=7cm]{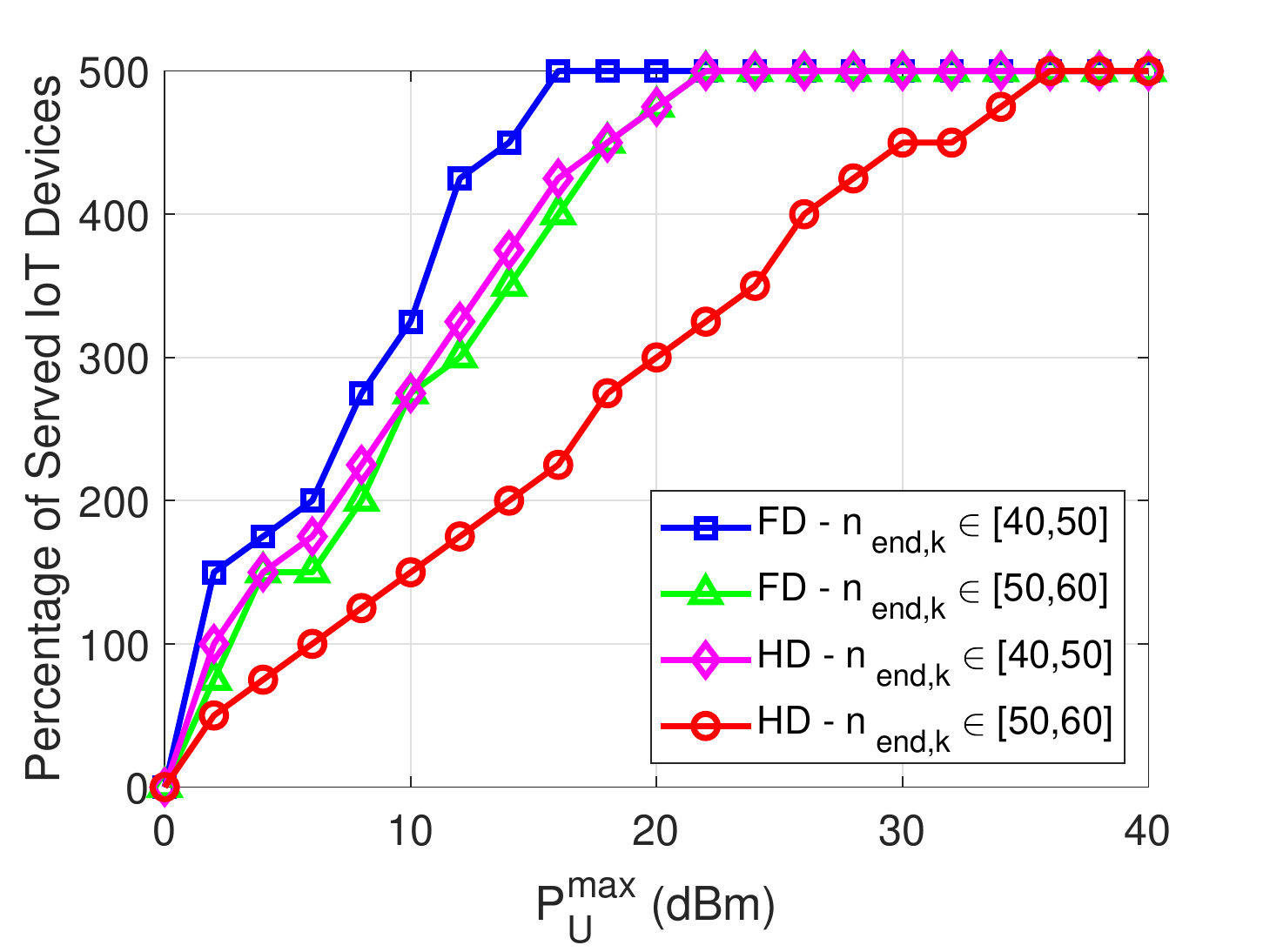}
	\caption{Percentage of served IoT devices vs. $P_{\rm U}^{\max
		}$ with different data size.}
	\label{fig:9}
\end{figure}
\begin{figure}
		\centering
		\includegraphics[width=9cm,height=7cm]{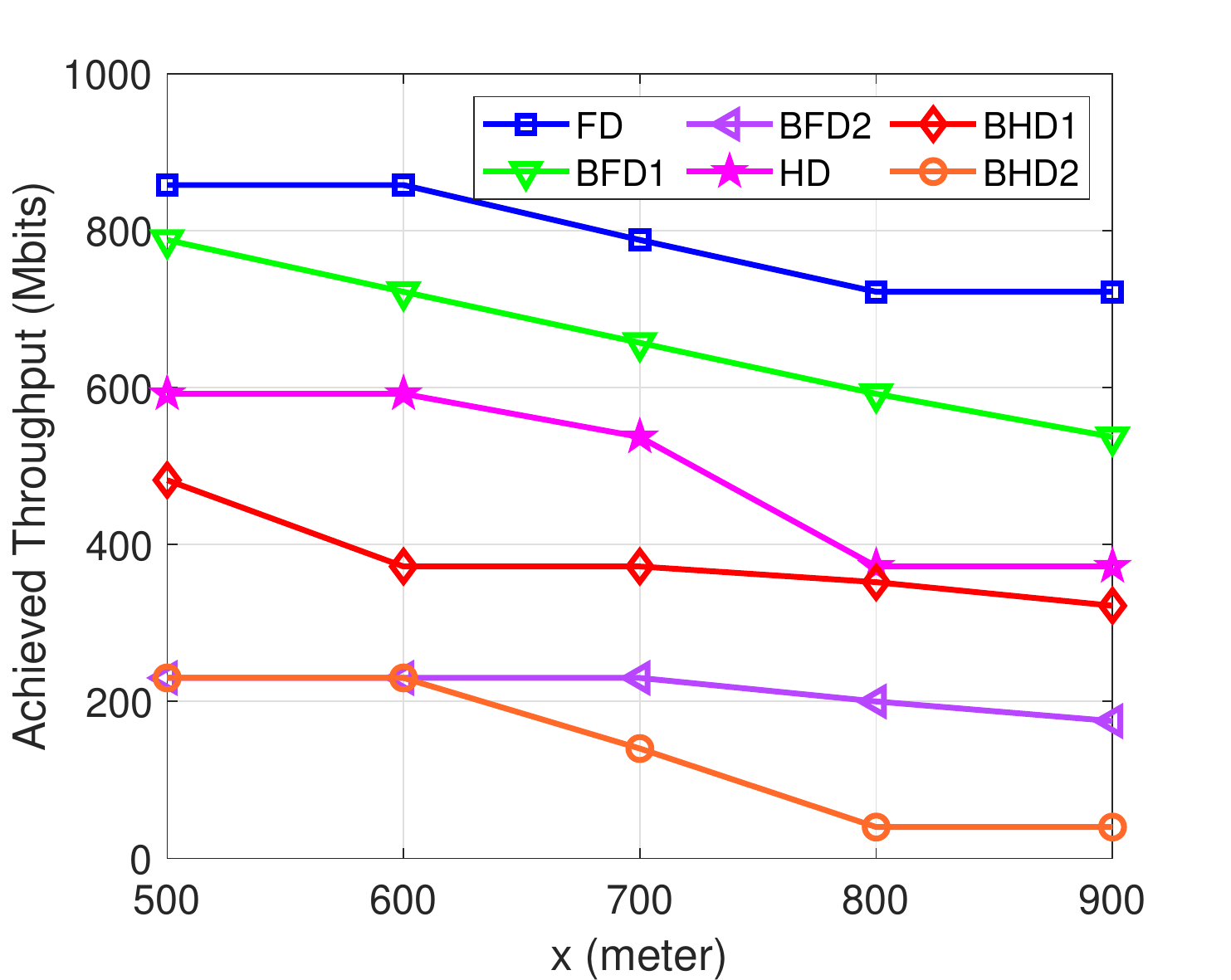}
		\caption{Total achievable throughput vs. different network sizes.}
		\label{fig:10}
\end{figure}
\begin{figure}
		\centering
		\includegraphics[width=9cm,height=7cm]{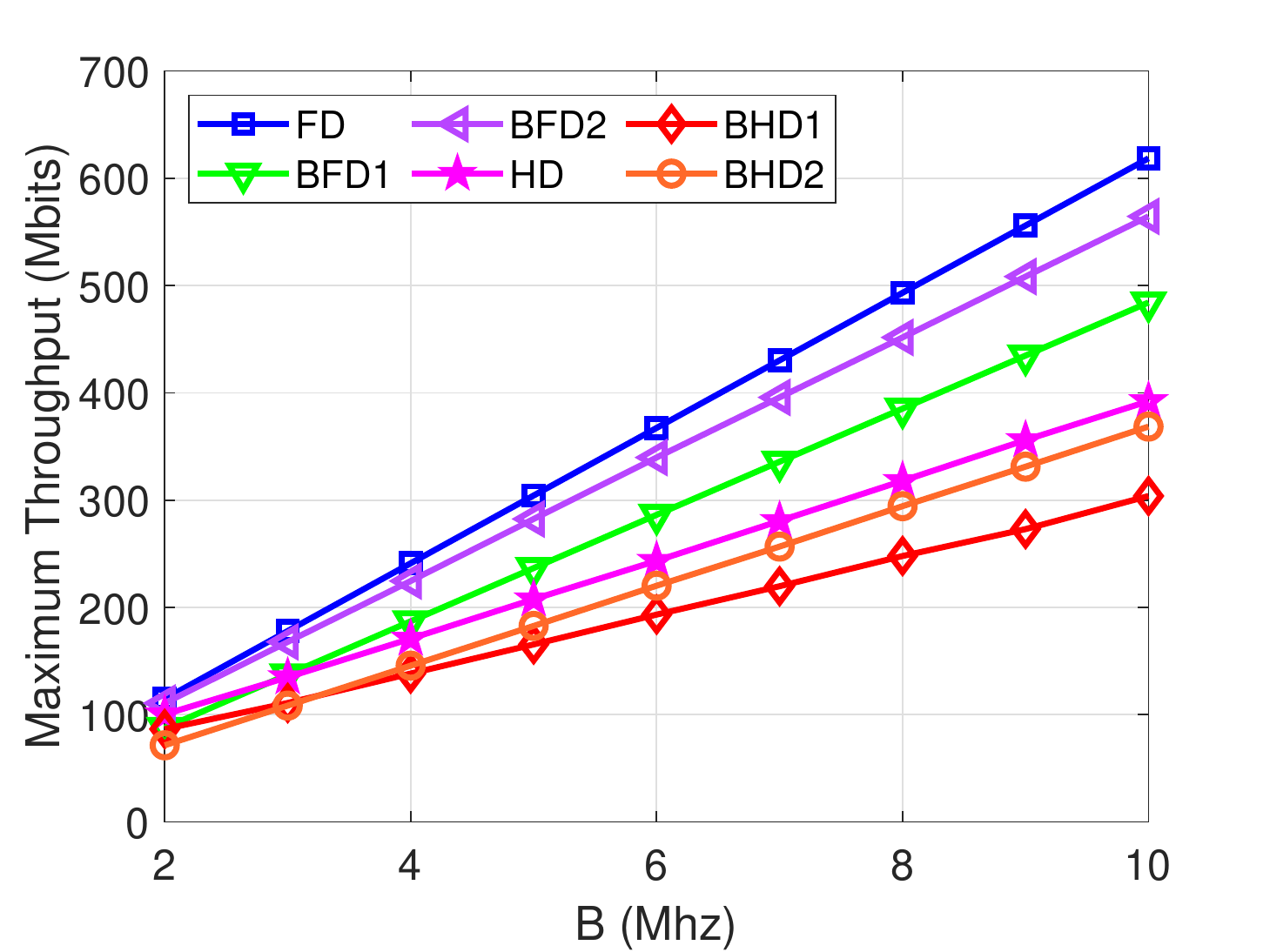}
		\caption{Maximum system throughput vs. different  bandwidth. }
		\label{fig:11}
\end{figure}

{In Fig. \ref{fig:7}, we investigate the effect of data size on system performance, where $K=$ 20, $B=5$ Mhz, $\eta_{{\rm start},k} \in [2,20]$, $n_{{\rm end}, k}^{ \min }= 30$ time slots, and $n_{{\rm end}, k}^{ \max }=55$ time slots, with $N$ = 70 time slots.} As inferred from the results, {the FD algorithm significantly improves the percentage of the served IoT devices compared to the HD algorithm for all values of cache size. Specifically, at $S_k \in$ [10, 30] Mbits and $C=$ 400 Mbits, the FD scheme can serve 85\% of IoT users on the network while  HD imposes 40\% of IoT users served. Furthermore, performance is degraded by increasing packet size $S_k$ due to limited available resources for IoT devices or the UAV, i.e., $P_{\rm U}^{\rm max}$, $P_k^{\rm max}$, $V_{\rm max}$, and $B$. Besides, when the $S_k$ value is small, corresponding to low data rate IoT devices, i.e., $S_k \in $ [1, 10] Mbits, the number of IoT users successfully served by proposed methods converge to a saturation value. Therefore, the UAV can operate in HD mode instead of FD one.}

{Fig. \ref{fig:8} illustrates the percentage of served IoT devices versus network size (maximum number of IoT devices located in the network area) with different data sizes, where $B=10$ MHz, $N=70$, $P_{\rm U}^{\rm max}=$ 18 dBm, $P_k^{\rm max}=$ 15 dBm, $n_{\rm start, k} \in [2,15]$ seconds, $n_{\rm end, k}^{ \min }= 25$ time slots, $n_{\rm end, k}^{ \max }=55$ time slots, and $C=$ 1000 Mbits. Similar to Figs. \ref{fig:3}-\ref{fig:7}, the percentage of IoT devices served by the FD method is better than the HD one. In addition, the percentage of served users is reduced by increasing the number of IoT users in the same network area. It is due to limited resources (i.e., bandwidth and transmit power allocated for UL and DL) and $V_{\rm max}$ when more IoT devices are considered. Besides, the percentage of served users will enlarge by decreasing the packet sizes $S_k$. This is expected because the UAV needs to spend more time and resources to compensate for higher $S_k$ increase.}

{Fig. \ref{fig:9} presents the results corresponding to the percentage of served GUs versus $P_{\rm U}^{\rm max}$ with different $n_{{\rm end},k}$ values. As shown, the number of served users is enhanced by increasing the power budget, i.e., $P_{\rm U}^{\rm max}$. Furthermore, FD scheme provides better results than HD scheme when $P_{\rm U}^{\rm max}$ is relatively small, e.g., $P_{\rm U}^{\rm max} < 22$ dBm with $n_{{\rm end},k} \in$ [40, 50] seconds. Nevertheless, the HD method can obtain the same number of served users as the FD method when the $P_{\rm U}^{\rm max}$ value is large, e.g., $P_{\rm U}^{\rm max} \ge 22$ dBm with $n_{{\rm end},k} \in$ [40, 50] seconds. This is because the FD mode suffers from RSI, which significantly increases the noise power in the UAV compared to the HD mode. In addition, RSI  is linearly proportional to $P_{\rm U}^{\rm max}$ as in \eqref{eq:10}. Therefore, when $P_{\rm U}^{\rm max}$ is large, the UAV should operate in HD mode since the FD mode requires more energy, which may exceed the system energy budget. It is due to the fact that in FD mode, the UAV starts to transmit data to GW earlier than in HD mode, which is highlighted in Fig. \ref{fig:5}. This results in higher energy consumption in the UAV when it manoeuvers in FD mode. }

\subsection{Throughput Maximization:}
{In the following, we present the corresponding results for the total throughput maximization problem described in Sections \ref{sec:FDrate} and \ref{sec:HDrate}. In Fig. \ref{fig:10}, the total achieved throughput is given as a function of network sizes, i.e., area is ranging from 500 m $\times$ 500 m to 900 m $\times$ 900 m, with $K=20$, $S_k$ is ranging from 20 to 70 Mbits, $B=10$ Mbits, $N=$ 70 time slots,} {$n_{{\rm start},k} \in [2,20]$ seconds, and $n_{{\rm end},k} \in [30,45]$ seconds. Specifically, the achieved throughput is defined as the total throughput that the UAV transfers from GUs to GW. Herein, we only take into account the throughput of successfully served GUs. We found that the proposed algorithms (i.e, FD and HD) significantly improve throughput performance compared to references (i.e., BFD1, BFD2, BHD1, BHD2) for all values of network sizes, i.e., $x$ (meters). Specifically, at $x$ = 700 m, FD algorithm can obtain 788 Mbits and BFD1 algorithm achieves less than 131 Mbits. Whereas BFD2, HD, BHD1, and BHD2  impose 230, 537, 372, and 140 Mbits, respectively. In particular, an interesting result is that HD is even better than BFD2, which underlines the superiority of the proposed algorithms over the references. That is due to the benefits of optimizing resource allocation.}

{In Fig. \ref{fig:11}, we investigate the effect of system bandwidth on maximum throughput, with $K=20$, $\rm area=$ 700 m $\times$ 700 m, $S_k$ ranging from 10 to 70 Mbits, $P_{\rm U}^{\rm max}=$ 18 dBm and $P_k^{\rm max}=$ 10 dBm, $N=$ 70 time slots, $n_{{\rm start},k} \in [2,20]$ seconds, and $n_{{\rm start},k} \in [45,55]$ seconds. Maximum throughput is defined as the total throughput that the UAV can convey to the GW regardless of whether or not each GU is successfully served. It has been observed that all schemes achieve better performance with an increase in total bandwidth. This is because the higher the bandwidth allocation, the greater the transmission can be achieved. Fig. \ref{fig:11} shows that FD schemes' performance is significantly better than the HD ones, since the UAV has more time to transfer collected data to GW in FD-based methods compared to HD-based ones. Therefore, they can be considered suitable for practical high throughput applications.}

{\section{Conclusion and Future Directions}}
\label{Sec:Conclusion}
We investigated the resource allocation and trajectory design for UAV-assisted FD IoT networks with the emergency communication system, {taking into account latency requirements of} IoT devices and the limited storage capacity of the UAV. In this context, we formulated {a novel problem} to maximize the total number of served IoT devices via a joint optimization of the UAV trajectory, allocated bandwidth, as well as the transmission power {of IoT devices} and UAV while satisfying the requested timeout constraints and storage capacity. Due to non-convexity of the formulated problem, {we first transformed} the original problem into a tractable form, {which is} then solved using an iterative algorithm with a polynomial computational complexity per iteration. Besides, pertaining to the realistic requirements for improving the estimation accuracy in a natural disaster or emergency scenario, we proposed an additional optimization problem in order to maximize the total collected data {while satisfying the threshold of a minimum number of served IoT devices.} We illustrated via numerical results that the proposed designs outperform the benchmark schemes in terms of both the total number of {served IoT devices }and the amount of collected data. Notably, in the scenarios such as when IoT devices' RT is not stringent, in the case of small data size, or required $P_{\rm U}^{\rm max}$ is large, the UAV should operate in the HD mode for a simple implementation.

{The outcome of this work will motivate future works in UAV-aided wireless systems. One possible problem is to extend this work to a multi-antenna UAV system, which imposes higher complexity but might further improve the network performance. Another promising problem is to consider low complexity yet efficient machine learning approaches to provide a reliable prediction of the LoS probability for any pair of UAV and GU locations, hence leading to enhance performance assurance.}


\section*{Appendix~A: Proof of Lemma~\ref{lemma:1}}
\label{Appendix:A}
\renewcommand{\theequation}{A.\arabic{equation}}
\setcounter{equation}{0}

{\underline{\textit{Proof for \eqref{eq:Lemma1_1} and \eqref{eq:Lemma1_2}}}: 
We consider a function $f(z)=\mathbb{E}_{Z}[\log_2(1+e^{\ln z})]$, $z > 0$. By adopting Jensen's inequality for convex function $\log_2(1+e^{\ln z})$, it yields
\begin{align}
	\label{eq:A1}
	f(z) \ge \log_2\big(1+ e^{\mathbb{E}_Z[\ln z]}\big).
\end{align} }

{Let us denote $Z \triangleq \Gamma_{\rm 1k} = \frac{ p_{1k}[n] | {\tilde h}_{1k}[n] |^2 \omega_0 } { {\big({H^2} + {{\big\| {\bq[n] - \bw_k} \big\|}^2}\big)^{\alpha/2}} \big(\phi^{\rm RSI} \sum\limits_{k^\ast \in {\cal K} \setminus k} p_{2k^\ast}[n] +  \sigma^2\big)} $. Thus, this is an
	exponentially distributed random variable with parameter $\lambda_Z \triangleq (\mathbb{E}[Z])^{-1} =  \frac{ \zeta_{1k}} { p_{1k}[n] \omega_0 } $ with $\zeta_{1k} \triangleq {\big({H^2} + {{\big\| {\bq[n] - \bw_k} \big\|}^2}\big)^{\alpha/2}}$ $\big(\phi^{\rm RSI} \sum\limits_{k^\ast \in {\cal K} \setminus k} p_{2k^\ast}[n] +  \sigma^2\big)$. By applying \cite[Eq. 4.331.1]{gradshteyn2014}, $\mathbb{E}_Z[\ln z]$ can be calculated as
\begin{align}
	\label{eq:A2}
	\mathbb{E}_Z[\ln z] &= \int_{0}^{+\infty} \lambda_Z e^{-z \lambda_Z } \ln{z} dz  = - \big(\ln (\lambda_Z) + E \big), \notag\\
	&= \ln \frac{ p_{1k}[n] \omega_0  } {\zeta_{1k}} -E, 
\end{align}
where $E$ is the Euler-Mascheroni constant, i.e., $E=0.5772156649$ as in \cite[Eq. 8.367.1]{gradshteyn2014}.}

{By substituting \eqref{eq:A2} into \eqref{eq:A1}, we obtain \eqref{eq:Lemma1_1}. Similar to \eqref{eq:A2}, we also easily achieve \eqref{eq:Lemma1_2} by adopting $Z \triangleq \Gamma_{\rm 2k}$.}

\section*{Appendix~B: Proof of Lemma~\ref{lemma:2}}
\label{Appendix:B}
\renewcommand{\theequation}{B.\arabic{equation}}
\setcounter{equation}{0}
As in \cite[Eq. (20)]{Dinh_1}, we have
\begin{align}
\label{eq:Lemma2_1}
h_1(x, y, z) \ge &\ln \left(1+\frac{x^{(j)}}{y^{(j)} z^{(j)}}\right) - \frac{x^{(j)}}{y^{(j)} z^{(j)}}   + 2 \frac{\sqrt{x^{(j)}}\sqrt{x}} {y^{(j)} z^{(j)}} \notag\\ &- \frac{x^{(j)}\left(x + yz\right)}{y^{(j)} z^{(j)}\left(x^{(j)} + y^{(j)} z^{(j)} \right) },\\
\label{eq:Lemma2_2}
h_2(x, z) \ge &\ln \left(1+\frac{x^{(j)}}{ z^{(j)}}\right) - \frac{x^{(j)}}{ z^{(j)}}  + 2 \frac{\sqrt{x^{(j)}}\sqrt{x}} { z^{(j)}} \notag\\ &- \frac{x^{(j)}\left(x + z\right)}{ z^{(j)}\left(x^{(j)}  + z^{(j)} \right) }.
\end{align}

By applying \eqref{eq:Lemma11}, the upper bound of $yz$ in \eqref{eq:Lemma2_1} is given by
\begin{align}
\label{eq:Lemma2_3}
yz \le  \frac{y^{(j)}}{2 z^{(j)}} z^2+\frac{z^{(j)}}{2 y^{(j)}} y^2,
\end{align}
with $x > 0,$ $y > 0,$ $z > 0,$ $x^{(j)} > 0,$ $y^{(j)} > 0,$ $z^{(j)} > 0$.

Then, substituting \eqref{eq:Lemma2_3} into \eqref{eq:Lemma2_1}, we obtain \eqref{eq:Lemma2_5} and  \eqref{eq:Lemma2_6}. {Lemma \ref{lemma:2} is hence completed.}

\section*{Appendix~C}
\label{Appendix:C}
\renewcommand{\theequation}{C.\arabic{equation}}
\setcounter{equation}{0}

From \eqref{eq:Lemma2_3}, the upper bound of $z_k[n] t_{1k}[n]$ in $r_{1k}^{\rm lb}[n]$ is:
\begin{align}
\label{eq:B1}
z_k[n] t_{1k}[n] &\le \left(z_{1k}[n] t_{1k}[n] \right)^{\rm ub} \notag\\ &\triangleq \frac{z_{1k}^{(j)}[n]\left(t_{1k}[n]\right)^2}{2 t_{1k}^{(j)}[n]}  + \frac{t_{1k}^{(j)}[n] \left(z_{1k}[n] \right)^2}{2z_{1k}^{(j)}[n]}.
\end{align}

By making use of  \eqref{eq:Lemma2_5}, \eqref{eq:Lemma2_6}, and \eqref{eq:B1}, the lower bound of $\Phi_{1k}[n]$ and $\Phi_{2k}[n]$ are, respectively

\begin{align}
\label{eq:B3}
\Phi_{1k}[n] &\ge \bar{\Phi}_{1k}[n] 
\triangleq B  \Big( \Xi_1 + \Xi_2-  \Xi_3 \Big), \\
\label{eq:B4}
\Phi_{2k}[n] &\ge \bar{\Phi}_{2k}[n] \triangleq B \Big( \Xi_4 + \Xi_5-  \Xi_6 \Big), 
\end{align}
where 
{\begin{IEEEeqnarray}{rCl} 
	&\Xi_1& \triangleq \log_2 \Biggl(1+ \frac{e^{-E} p_{1k}^{(j)}[n] \omega_0 } {z_{1k}^{(j)}[n] t_{1k}^{(j)}[n] } \Biggr)  - \frac{e^{-E} p_{1k}^{(j)}[n] \omega_0} {z_{1k}^{(j)}[n] t_{1k}^{(j)}[n] \ln 2 }, \nonumber\\ \vspace{-0.01cm}
	&\Xi_2& \triangleq e^{-E} \omega_0  \frac{2 \sqrt{p_{1k}^{(j)}[n]} \sqrt{p_{1k}[n]}} {z_{1k}^{(j)}[n] t_{1k}^{(j)}[n] \ln 2 } , \notag \\  \vspace{-0.01cm}
	&\Xi_3& \triangleq \frac{e^{-E} p_{1k}^{(j)}[n] \omega_0}{\big(e^{-E} p_{1k}^{(j)}[n] \omega_0 + z_{1k}^{(j)}[n] t_{1k}^{(j)}[n] \big) z_{1k}^{(j)}[n] t_{1k}^{(j)}[n] \ln 2} \notag \\   \vspace{-0.01cm} &\times& \Bigg(e^{-E} p_{1k}[n] \omega_0 + \frac{z_{1k}^{(j)}[n]\left(t_{1k}[n]\right)^2 } {2 t_{1k}^{(j)}[n] } + \frac{   t_{1k}^{(j)}[n] \left(z_{1k}[n] \right)^2 }{2z_{1k}^{(j)}[n] } \Bigg),
 \nonumber	 \\  \vspace{-0.01cm}
&\Xi_4& \triangleq \log_2 \Bigg(1+ \frac{e^{-E} p_{2k}^{(j)}[n] \omega_0 }{z_{2k}^{(j)}[n] \sigma^2 } \Bigg) - \frac{e^{-E} p_{2k}^{(j)}[n] \omega_0 }{z_{2k}^{(j)}[n] \sigma^2 \ln 2}, \nonumber\\  \vspace{-0.01cm}
&\Xi_5& \triangleq \frac{ e^{-E} \omega_0 }{z_{2k}^{(j)}[n] \sigma^2 \ln 2}    2\sqrt{ p_{2k}^{(j)}[n]} \sqrt{ p_{2k}[n]} , \nonumber\\   \vspace{-0.01cm}
&\Xi_6& \triangleq  \frac{ e^{-E} p_{2k}^{(j)}[n] \omega_0  }{   e^{-E} p_{2k}^{(j)}[n] \omega_0 + z_{2k}^{(j)}[n] \sigma^2 } 
\times  \frac{\Big(e^{-E} p_{2k}[n] \omega_0  + z_{2k}[n] \sigma^2 \Big)} {z_{2k}^{(j)}[n] \sigma^2 \ln 2}. \notag
\end{IEEEeqnarray}}

\section*{Appendix~D: Proof of Proposition~\ref{proposition_1}}
\label{Appendix:D}
\renewcommand{\theequation}{D.\arabic{equation}}
\setcounter{equation}{0}

For the sake of notational convenience, let us define the feasible set $\chi^{(j)}$ of \eqref{eq:P13} at the initial stage of the $(j+1)$-th iteration, such that
\begin{align}
	\vspace{-0.01cm}
	\chi^{(j)} \triangleq \{\boldsymbol{\Psi}^{(j)}| \text{s.t. \eqref{eq:P13:b}-\eqref{eq:P13:g} are feasible } \}.
\end{align}

First, we recall that the approximate functions presented in Section \ref{Sec:3} satisfy properties of IA algorithm  \cite{marks,beck2010}. Let ${\mathbb{F}}(\boldsymbol{\Psi})$ and $\widetilde{\mathbb{F}}(\boldsymbol{\Psi})$ denote the objective function of \eqref{eq:P11} and \eqref{eq:P13}, respectively. Following IA principles, the feasible region of approximated convex function \eqref{eq:P13} is a subset of the feasible region of relaxed problem \eqref{eq:P11} \cite[Property i of Lemma 2.2]{beck2010}. Thus, it is true that
\begin{align}
\tag{D.1}
{\mathbb{F}}(\boldsymbol{\Psi}) &\ge \widetilde{\mathbb{F}}(\boldsymbol{\Psi}), \; \forall \boldsymbol{\Psi}, \\
 	\tag{D.2}
 	{\mathbb{F}}(\boldsymbol{\Psi}^{(j)}) &= \widetilde{\mathbb{F}}(\boldsymbol{\Psi}^{(j)}), \; \forall \boldsymbol{\Psi}. 
 \end{align}
Thus, it follows that
\begin{align}
\tag{D.3}
{\mathbb{F}}(\boldsymbol{\Psi}^{(j+1)}) \geq \widetilde{\mathbb{F}}(\boldsymbol{\Psi}^{(j+1)}) 
 \geq \widetilde{\mathbb{F}}(\boldsymbol{\Psi}^{(j)})
= {\mathbb{F}}(\boldsymbol{\Psi}^{(j)}),
\end{align}
where the first inequality is due to (D.1). The second inequality  is attributed to the fact that  $\boldsymbol{\Psi}^{(j+1)}$ is a better solution for {\eqref{eq:P11}} than $\boldsymbol{\Psi}^{(j)}$ \cite[Property iv of Lemma 2.2]{beck2010}. Moreover, the sequence $\{{\mathbb{F}}(\boldsymbol{\Psi}^{(j)})\}$ will converge, as  shown in \cite[Corollary 2.3 ]{beck2010}, and each accumulation point $\boldsymbol{\Psi}^\star$ of the sequence $\{\boldsymbol{\Psi}^{(j)}\}$ is a Karush-Kuhn-Tucker point as in \cite[Theorem 1]{marks} and \cite[Proposition 3.2]{beck2010}. Furthermore, since the feasible set $\chi^{(j)}$ is a convex connected set due to the convexity of \eqref{eq:P13} \cite{boyd2002}. Moreover, it is closed and bounded because of power constraints \eqref{eq:P1:l} and \eqref{eq:P1:m}, bandwidth constraints \eqref{eq:P1:h} and \eqref{eq:P1:i}, and limited flying time. Consequently, we can obtain a locally optimal solution to \eqref{eq:P11} according to \cite[Corollary 1]{marks}, which completes the proof.

\balance
\bibliographystyle{IEEEtran}
\bibliography{IEEEfull}
\begin{IEEEbiography}[{\includegraphics[width=1in,height=1.25in,clip,keepaspectratio]{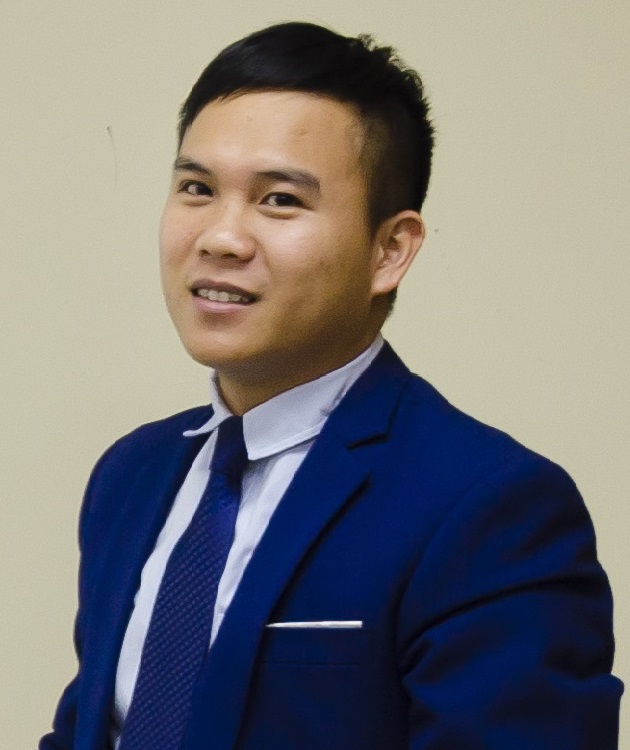}}]
	{Dinh-Hieu Tran} (S'20)  was born and grew up in Gia Lai, Vietnam (1989). He received the B.E. degree in Electronics and Telecommunication Engineering Department from Ho Chi Minh City University of Technology, Vietnam, in 2012. In 2017, he finished the M.Sc degree in Electronics and Computer Engineering from Hongik University (Hons.), South Korea. He is currently pursuing the Ph.D. degree at the Interdisciplinary Centre for Security, Reliability and Trust (SnT), University of Luxembourg, under the supervision of Prof. Symeon Chatzinotas and Prof. Bj\"orn Ottersten. His research interests include UAVs, IoTs, Mobile Edge Computing, Caching, Backscatter, B5G for wireless communication networks.  He was a recipient of the IS3C 2016 best paper award.
\end{IEEEbiography}

\begin{IEEEbiography}[{\includegraphics[width=1in,height=1.25in,clip,keepaspectratio]{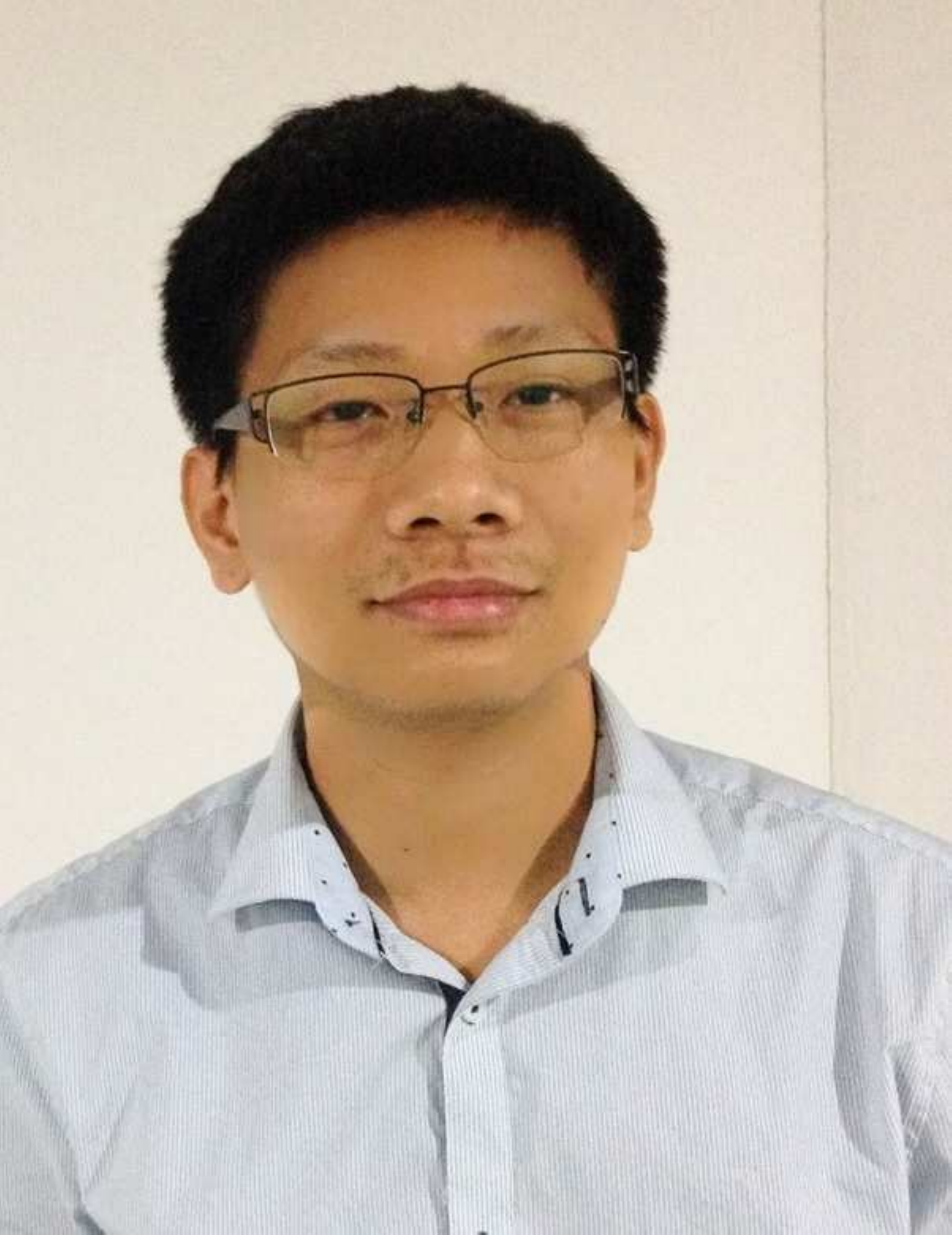}}]
	{Van-Dinh Nguyen} (Member, IEEE) received the B.E. degree in electrical engineering from the HoChi Minh City University of Technology, Vietnam, in 2012, and the M.E. and Ph.D. degrees in electronic engineering from Soongsil University, Seoul, South Korea, in 2015 and 2018, respectively.

	He was a Post-Doctoral Researcher and a Lecturer with Soongsil University, a Post-Doctoral Visiting Scholar with the University of Technology Sydney, Australia, from July 2018 to August 2018, and a Ph.D. Visiting Scholar with Queen’s University Belfast, U.K., from June 2015 to July 2015 and in August 2016. He is currently a Research Associate with the Interdisciplinary Centre for Security, Reliability and Trust (SnT), University of Luxembourg. His current research interests include fog/edge computing, the Internet of Things, 5G networks, and machine learning for wireless communications. He received several best conference paper awards, IEEE TRANSACTIONS ON COMMUNICATIONS Exemplary Reviewer 2018, and IEEE GLOBECOM Student Travel Grant Award 2017. He has authored or co-authored 40 papers published in international journals and conference proceedings. He has served as a reviewer for many top-tier international journals on wireless communications and has also been a technical programme committee member for several flag-ship international conferences in the related fields. He is an Editor of the IEEE OPEN JOURNAL OF THE COMMUNICATIONS SOCIETY and IEEE COMMUNICATIONS LETTERS 
\end{IEEEbiography}

\begin{IEEEbiography}[{\includegraphics[width=1in,height=2in,clip,keepaspectratio]{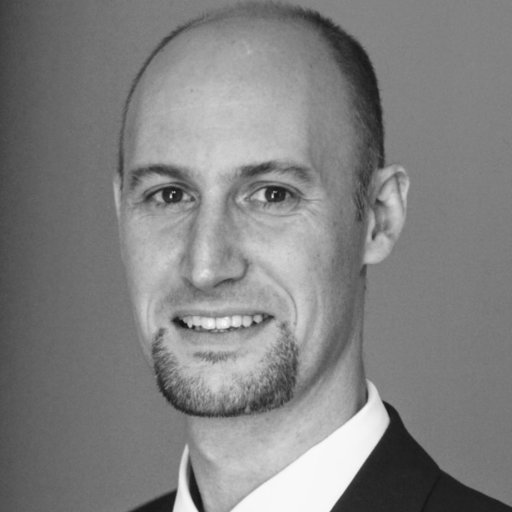}}]
	{Symeon Chatzinotas}, (S'06-M'09-SM'13) is currently Full Professor / Chief Scientist I in Satellite Communications and Head of the SIGCOM Research Group at SnT, University of Luxembourg. He is coordinating the research activities on communications and networking, acting as a PI for more than 20 projects and main representative for 3GPP, ETSI, DVB.
	In the past, he has been a Visiting Professor at the University of Parma, Italy, lecturing on “5G Wireless Networks”. He was involved in numerous R$\&$D projects for NCSR Demokritos, CERTH Hellas and CCSR, University of Surrey.
	He was the co-recipient of the 2014 IEEE Distinguished Contributions to Satellite Communications Award and Best Paper Awards at EURASIP JWCN, CROWNCOM, ICSSC. He has (co-)authored more than 450 technical papers in refereed international journals, conferences and scientific books.
	He is currently in the editorial board of the IEEE Transactions on Communications, IEEE Open Journal of Vehicular Technology and the International Journal of Satellite Communications and Networking.
\end{IEEEbiography}

\begin{IEEEbiography}[{\includegraphics[width=1in,height=1.25in,clip,keepaspectratio]{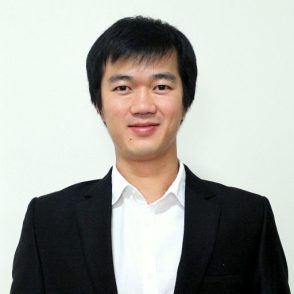}}]
	{Thang X. Vu} (M'15) was born in Hai Duong, Vietnam. He received the B.S. and the M.Sc., both in Electronics and Telecommunications Engineering, from the VNU University of Engineering and Technology, Vietnam, in 2007 and 2009, respectively, and the Ph.D. in Electrical Engineering from the University Paris-Sud, France, in 2014. 
	
	In 2010, he received the Allocation de Recherche fellowship to study Ph.D. in France. From September 2010 to May 2014, he was with the Laboratory of Signals and Systems (LSS), a joint laboratory of CNRS, CentraleSupelec and University Paris-Sud XI, France. From July 2014 to January 2016, he was a postdoctoral researcher with the Information Systems Technology and Design (ISTD) pillar, Singapore University of Technology and Design (SUTD), Singapore. Currently, he is a research scientist at the Interdisciplinary Centre for Security, Reliability and Trust (SnT), University of Luxembourg, Luxembourg. His research interests are in the field of wireless communications, with particular interests of 5G networks and beyond, machine learning for communications and cross-layer resources optimization. He was a recipient of the SigTelCom 2019 best paper award.
\end{IEEEbiography}

\begin{IEEEbiography}[{\includegraphics[width=1in,height=1.25in,clip,keepaspectratio]{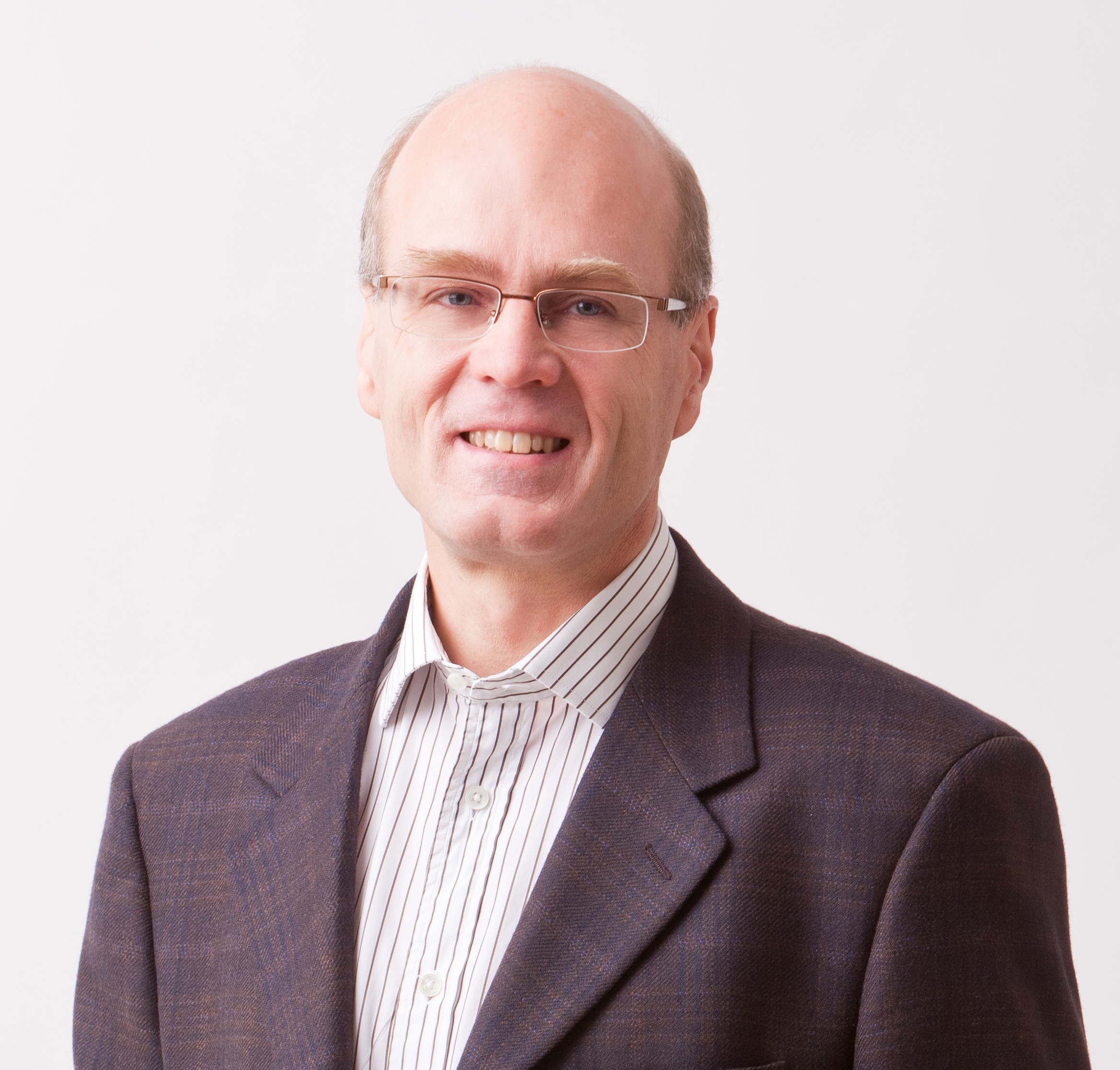}}]
	{Bj\"orn Ottersten}, (S'87-M'89-SM'99-F'04) was born in Stockholm, Sweden, in 1961. He received
	the M.S. degree in electrical engineering and applied physics from Linköping University, Linköping, Sweden, in 1986, and the Ph.D. degree in electrical engineering from Stanford University, Stanford, CA, USA, in 1990. He has held research positions with the Department of Electrical Engineering, Linköping
	University, the Information Systems Laboratory, Stanford University, the Katholieke Universiteit	Leuven, Leuven, Belgium, and the University of Luxembourg, Luxembourg. From 1996 to 1997, he was the Director of Research with ArrayComm, Inc., a start-up in San Jose, CA, USA, based on his patented technology. In 1991, he was appointed Professor of signal processing with the Royal Institute of Technology (KTH), Stockholm, Sweden. Dr. Ottersten has been Head of the Department for Signals, Sensors, and Systems, KTH, and Dean of the School of Electrical Engineering, KTH. He is currently the Director for the Interdisciplinary Centre for Security, Reliability and Trust, University of Luxembourg.
	He is a recipient of the IEEE Signal Processing Society Technical Achievement Award and the European Research Council advanced research grant twice. He has co-authored journal papers that received the IEEE Signal Processing Society Best Paper Award in 1993, 2001, 2006, 2013, and 2019, and 8 IEEE conference papers best paper awards. He has been a board member of IEEE Signal Processing Society, the Swedish Research Council and currently serves of the boards of EURASIP and the Swedish Foundation for Strategic Research. He has served as an Associate Editor for the IEEE TRANSACTIONS ON SIGNAL PROCESSING and the Editorial Board of the IEEE Signal Processing Magazine. He is currently a member of the editorial boards of IEEE Open Journal of Signal Processing, EURASIP Signal Processing Journal, EURASIP Journal of Advances Signal Processing and Foundations and Trends of Signal Processing. He is a fellow of EURASIP.	
\end{IEEEbiography}	
\end{document}